\documentclass[10pt,journal,compsoc,table]{IEEEtran}

\usepackage{epsfig}
\usepackage{epstopdf}
\usepackage{times}
\usepackage{framed}
\usepackage{enumitem}
\usepackage{float}
\usepackage{algorithm}
\usepackage{etoolbox}
\usepackage{amssymb}
\usepackage{amsmath}
\usepackage{subfigure}
\usepackage{multirow}
\usepackage{url}
\usepackage{lipsum}
\usepackage{soul}
\usepackage{hyperref}
\usepackage{mathtools}
\usepackage{bbm}
\usepackage{booktabs}
\usepackage{xcolor}
\usepackage{makecell}
\usepackage{pifont}
\usepackage{flushend}

\newcommand{\xmark}{\ding{53}}%
\newcommand{\bbcheckmark}{\ding{52}}

\newcommand\bundermat[2]{%
  \makebox[0pt][l]{$\smash{\underbrace{\phantom{%
    \begin{matrix}#2\end{matrix}}}_{\text{#1}}}$}#2}

\def\1{\mathbf{1}}
\newcommand{\vv}[1]{\boldsymbol{#1}}
\newcommand{\vvi} [2]{ \boldsymbol{#1}^{(#2)} }

\newcounter{subeqn} %
\makeatletter
\@addtoreset{subeqn}{equation}
\makeatother

\usepackage{array} 
\def\ignore#1\endignore{}
\newcolumntype{h}{@{}>{\ignore}l<{\endignore}} 

\newcolumntype{x}[1]{%
>{\centering\hspace{0pt}}p{#1}}%

\usepackage{latexsym}
\usepackage{multicol}
\usepackage{eurosym}
\usepackage{amsthm}
\usepackage{xspace}
\usepackage{pdfpages}
\usepackage{verbatim}
\usepackage{stfloats}

\newenvironment{sketch}{{\noindent \it Sketch of Proof:~}}

\newtheorem{problem}{Problem}
\newtheorem{lemma}{Lemma}

\newtheorem{theorem}{Theorem}

\newcommand{\name}{SARDO}

\title{\name{}: An Automated Search-and-Rescue Drone-based Solution for Victims Localization}
\author{Antonio~Albanese,~\IEEEmembership{Student Member,~IEEE,}
Vincenzo~Sciancalepore,~\IEEEmembership{Senior Member,~IEEE,}
Xavier~Costa-P\'erez,~\IEEEmembership{Senior Member,~IEEE}
\thanks{\textit{A. Albanese is with NEC Laboratories Europe and University Carlos III of Madrid, 28911 Legan\'es, Spain.\newline
V. Sciancalepore and X. Costa-P\'erez are with NEC Laboratories Europe, 69115 Heidelberg, Germany. \newline
Emails: \{antonio.albanese, vincenzo.sciancalepore, xavier.costa\}@neclab.eu}.}%
}

\begin{document}

\IEEEtitleabstractindextext{
\begin{abstract}
Natural disasters affect millions of people every year. Finding missing persons in the shortest possible time is of crucial importance to reduce the death toll. This task is especially challenging when victims are sparsely distributed in large and/or difficult-to-reach areas and cellular networks are down.


In this paper we present \name, 
a \emph{drone-based} search and rescue solution that exploits the high penetration rate of mobile phones in the society to localize missing people. \name~is an autonomous, all-in-one drone-based mobile network solution that does not require infrastructure support or mobile phones modifications. It builds on novel concepts such as \emph{pseudo-trilateration} combined with machine-learning techniques to efficiently locate mobile phones in a given area. Our results, with a prototype implementation in a field-trial~\cite{onlinevideo}, show that \name{} rapidly  
determines the location of  mobile phones ($\sim 3$ min/UE) in a given area with an accuracy of few tens of meters and at a low battery consumption cost ($\sim 5\%$).

State-of-the-art localization solutions for disaster scenarios rely either on mobile infrastructure support or exploit onboard cameras for human/computer vision, IR, thermal-based localization. To the best of our knowledge, \name~is the first \emph{drone-based cellular search-and-rescue solution} able to accurately localize missing victims through mobile phones. 


%
\end{abstract}


}
\maketitle

\thispagestyle{empty}	

\section{Introduction}
\label{s:intro}

In 2017, $335$ natural disasters affected over $95.6$ million people, killing an additional $9,697$ and costing a total of $335$ billion USD. 
 The deadliest event in 2017 was the landslide in Sierra Leone in August, with $1102$ reported dead or missing, followed by Cyclone Okchi in December with $884$ reported dead or missing in India. Notably, these two events are characterized by a high number of \emph{missing people}, representing over half of the total death toll \cite{disaster_report}.
In addition to natural disasters, human-generated threats (e.g. fires, electrical outages, terrorism) might as well 
require solutions to improve first responders' capabilities to address such situations. 


One of the most compelling challenges when a disaster strikes is to quickly establish a first contact with affected victims that might be trapped or hidden from rescue teams. The response to such a dire situation shall be prompt and effective even when the terrestrial communication network is down, the debris makes user GNSS information not available or victims are physically incapable to transmit their current location, thereby exacerbating the search and rescue procedure. 

\emph{Unmanned Aerial Vehicles} (\emph{UAVs}) or \emph{Drones} have recently emerged as a cost-efficient alternative to address emergency scenarios~\cite{flying_cell} for multiple reasons. First, UAVs can be rapidly deployed in disaster areas providing on-demand mobile networks.
Second, UAVs may rapidly approach difficult-to-reach locations, such as mountains, deserts, or devastated areas and cover large search areas with sparse victims distribution. 
Finally, given the high penetration rate of mobile devices in our society, it can be reasonably assumed that victims are equipped with smart devices, e.g., smart phones and wearables, that can be detected by UAV mobile networks.
In this paper, we present our Search-And-Rescue DrOne-based solution, \name{}, an all-in-one localization system that supports first responders to quickly identify and localize victims in disaster areas.
\name{} $i$) scans and spots target users by means of an integrated IMSI-catcher, $ii$) applies machine-learning principles on the distance measurements 
to perform our novel pseudo-trilateration localization technique, $iii$) relies on a neural network to predict future target positions, $iv$) closes the feedback loop with a control component that automatically adjusts the UAV trajectory for improving the localization accuracy. \name{} has been implemented and tested in a field-trial scenario 
~\cite{onlinevideo} with COTS components. Our results prove the feasibility of the solution and provide quantitative system performance figures.

\section{\name~Framework overview} \label{s:framework}


State-of-the-art \emph{localization} solutions rely on \emph{trilateration} and \emph{triangulation} methods, which combine measurements collected by different anchors.
In contrast, hereafter we present \name{} that aims to find victims' locations by keeping track of their mobile phone signals in disaster areas with the information collected by a \emph{single} UAV that sweeps a given predefined area and acts as a \emph{portable cellular base station}. The \emph{building blocks} of \name{} are depicted in~Fig.~\ref{fig:block_diagram}. 
 
\textbf{Time-of-flight measurement process.} A UAV equipped with a light-weight base station scans on a predefined disaster area to discover victims\footnote{Our analysis focuses on single-target detection. However, \name{} can be used to locate multiple targets in a sequential manner or can be readily extended to locate users' clusters by means of a customized implementation.}. No protocol stack modifications are introduced such that compatibility with commercial mobile phones is guaranteed, as described in Section~\ref{s:implementation}. This allows measuring the time of flight (ToF) of a user uplink signal that is fed into the novel \emph{pseudo-trilateration} algorithm, as detailed in Section~\ref{s:pseudo-tri}.
 
\textbf{User position estimate.} A single view-point, the UAV, exploits \emph{time diversity} to retrieve different user ToF values that are combined to estimate the current user position as described in Section~\ref{s:aicnn}. In addition, when the user moves an estimate of its motion trajectory is derived by means of a Convolutional Neural Network (CNN) that extracts features from different ToF values. Such features are processed by the output deep Feed-Forward Neural Network (FFNN) that learns and implements the concept of pseudo-trilateration. 

\textbf{Future user positions prediction.} The returned set of previous user positions is used to predict future locations. A Long Short-Term Memory (LSTM) neural network keeps memory of previous system states and forecasts a set of future user positions. This is explained in detail in Section~\ref{s:lstm}.

\textbf{UAV Relocation.} A large prediction time window results in a lower accuracy of the obtained prediction. Thus, we implement a component that leverages control theory to adjust the next UAV motion trajectory. On the one hand, a high accuracy in the forecasting process fosters the UAV to focus on the expected user location and reduce the scanning area. On the other hand, detected errors in the prediction process forces the UAV to enlarge the scanning area in order to recover from previous wrong decisions, as shown in Section~\ref{s:controller}.

The execution of the entire process requires some time. Therefore, we design the system as a sequence of atomic localization tasks executed within a predefined short time threshold, e.g., tens of seconds. Upon completion of each task, a new set of inputs is fed to the UAV that quickly relocates and starts the probing process. Intuitively, the variable duration of each task impacts on the number of collected measurements that, in turn, drive the accuracy of the pseudo-trilateration (and the forecasting phase). Thus, we properly design \name{} to accurately provide users position within few minutes and explain the setting parameters in detail in Section~\ref{s:results}.

\begin{figure}[t!]
      \centering
      \includegraphics[trim={3mm 3.3cm 7mm 3.4cm}, clip, width=\linewidth ]{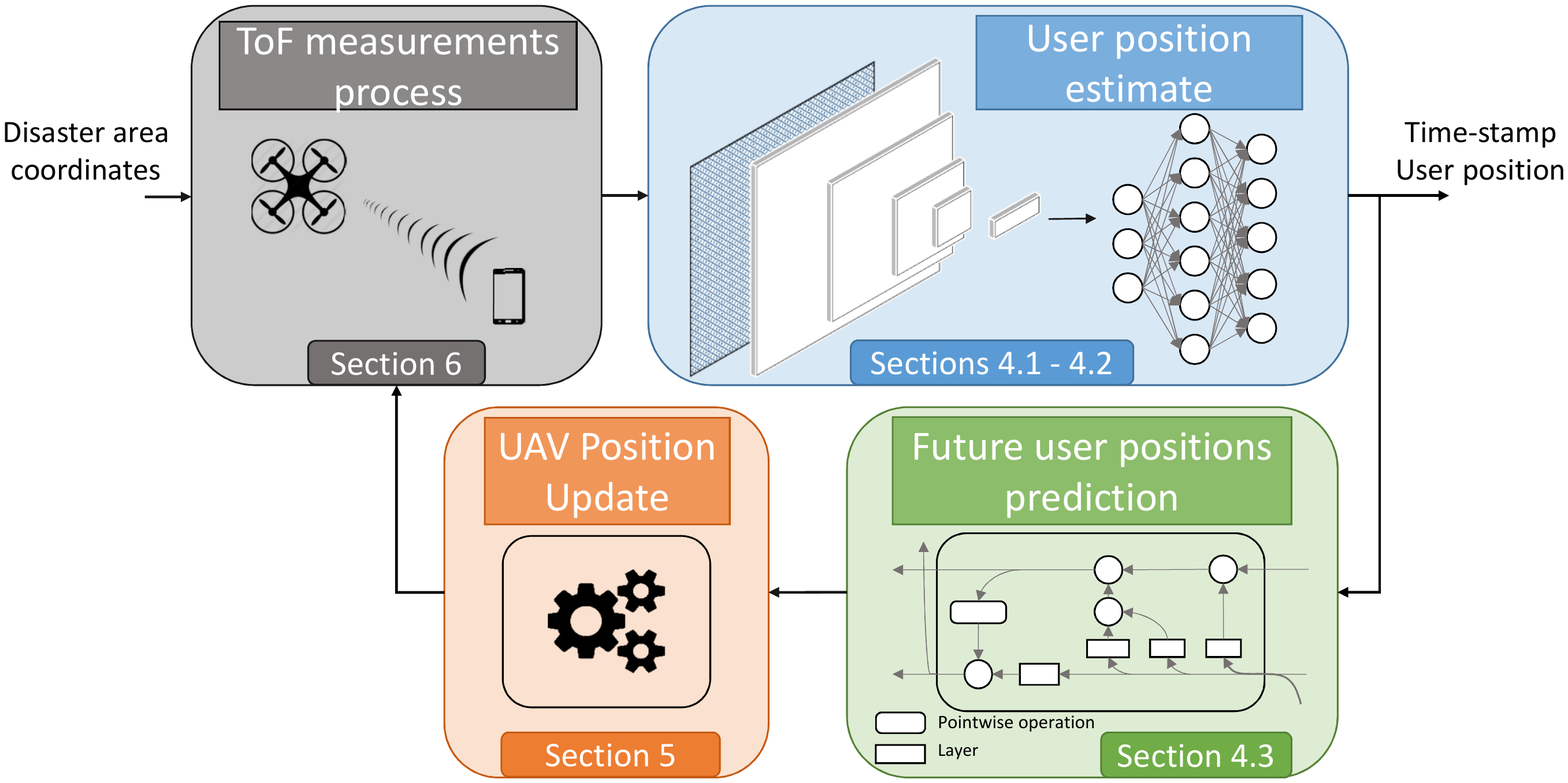}
      \caption{Overview of \name{}'s Building Blocks}
      \label{fig:block_diagram}
\end{figure}

\section{Geo-Localization Model} \label{s:motivation}

The simplest localization technique proposed in the literature for static objects requires distinct reference points, dubbed as anchors, to measure their distances from the target. It is known as trilateration, if $2$-dimensional coordinates are sought, or multi-lateration technique, in case of $3$-dimensional coordinates\footnote{Note that the terms ``trilateration'' and ``multilateration'' are used interchangeably throughout the paper.}.  However, these techniques suffer from major limitations when the measurements get affected by noise and/or mobility. In the following, we describe the multi-lateration technique pointing out its limitations and then we introduce our novel technique, namely \emph{pseudo-trilateration}, that is robust to noisy measurements and mobile objects.

\subsection{Legacy multi-lateration}
\label{s:legacy-multi}

Commercial localization systems, such as GNSS, effectively address noise and mobility issues providing an accuracy up to $5$ meters~\cite{SPECHT2018}. However, they are developed and managed by third (untrusted) parties fostering telco operators to deploy and control their own reliable localization solutions. On the other hand, cellular networks---where user localization techniques are implemented assuming base stations acting as anchors---may provide a localization accuracy up to hundreds of meters~\cite{Leontiadis2014}. This is mostly due to the fact that LTE and 4G systems leverage the Received Signal Strength Indicator (RSSI) that roughly provides an estimation of the distance based on pre-defined fingerprints. Additionally, for mobile users only few base stations---the ones covering the area wherein the user moves---may retrieve useful distance information thereby requiring full cooperation of a high number of anchors in the area in order to jointly calculate the instantaneous user position. Ideally, a full deployment of base stations within the considered area would minimize the localization error but, at the same time, increase the overhead due to the combination of different measurements. 

Such localization solutions apply the well-known concept of multi-lateration that involves a number of viewpoints (more than four in case of $3$-dimensional geo-localization) simultaneously measuring the distance from the target. Analytically, we assume $I=|\mathcal{I}|$ anchors where each anchor position $\vv{x}_i=\{x_i,y_i,z_i\},\forall i\in\mathcal{I}$ is known a-priori. We can jointly compute the position of the target object $\vv{x}_t=\{x_t,y_t,z_t\}$ solving the following set of equations
\begin{equation}
\label{eq:multi-basis}
(x_t-x_i)^2+(y_t-y_i)^2+(z_t-z_i)^2 = \gamma_i^2, \quad\forall i\in\mathcal{I},
\end{equation}
where the distance $\gamma_i$ between the anchor $i$ and the target object $t$ identifies the radius of a sphere centered in $\vv{x}_i$ and passing through $\vv{x}_t$. If the retrieved distances are not biased, the above set of equations admits a unique solution $\vv{x}_t$ represented by the intersection of all involved spheres. However, this result might not hold in real wireless environments due to the presence of channel fading where the measured distance---denoted as $\hat{\gamma_i}$---is calculated through Time of Flight (ToF) or Time of Arrival (ToA)~\cite{tof}. Indeed, this might lead to multiple intersection points between any possible pair of spheres and require an approximation on the location of the target. We can rewrite Eq.~\eqref{eq:multi-basis} as a linear system of $(I-1)$ equations in matrix form $S\vv{x}=\vv{p}$, where
\begin{equation}
\nonumber
\label{eq:matr_form} 
S= \begin{bmatrix}
x_I-x_1 & y_I-y_1 & z_I-z_1 \\
x_I-x_2 & y_I-y_2 & z_I-z_2 \\
\vdots & \vdots & \vdots \\
x_I-x_{I-1} & y_I-y_{I-1} & z_I-z_{I-1}
\end{bmatrix},
\vv{x} = \begin{bmatrix}
x_t\\
y_t\\
z_t
\end{bmatrix} 
\end{equation}

and 
\begin{equation}
\nonumber
\label{eq:matr_form2}
\vv{p}= \begin{bmatrix}
(\hat{\gamma}_1^2-\hat{\gamma}_I^2)-(x_1^2-x_I^2)-(y_1^2-y_I^2)-(z_1^2-z_I^2) \\
(\hat{\gamma}_2^2-\hat{\gamma}_I^2)-(x_2^2-x_I^2)-(y_2^2-y_I^2)-(z_2^2-z_I^2) \\
\vdots \\
(\hat{\gamma}_{I-1}^2\!\!-\hat{\gamma}_I^2)\!-\!(x_{I-1}^2\!\!-\!\!x_I^2)\!-\!(y_{I-1}^2\!\!-\!\!y_I^2)\!-\!(z_{I-1}^2\!\!-\!\!z_I^2)
\end{bmatrix}.
\end{equation}

Let us denote $\vv{x}^*$ the approximated target position due to the measurements noise, namely $\vv{x}^*$ satisfying $S\vv{x^*}\approx \vv{p}$. Finding the derivative of the sum of the squares of the residuals yields that $S^TS\vv{x}^*=A^T\vv{p}$.
While a non-singular $S^TS$ exhibits $\vv{x}^* = (S^TS)^{-1}S^T\vv{p}$ as unique solution, we apply the non-linear least squares method when it is close to singular that minimizes the sum of the squares of the errors on the distances.

Let us denote the function to minimize as the following
\begin{equation}
\label{eq:optfunction}
F(x,y,z) = \sum\limits_{i=1}^I (\gamma_i - \hat{\gamma_i})^2 = \sum\limits_{i=1}^I f_i(x,y,z)^2,
\end{equation}
where
\begin{equation}
\label{eq:function}
f_i(x,y,z)\!\!=\!\!\gamma_i - \hat{\gamma_i}\!=\!\sqrt{(x_t-x_i)^2\!+\!(y_t-y_i)^2\!+\!(z_t-z_i)^2}-\hat{\gamma_i}.
\end{equation}
We can calculate the vector $\vv{g}=[\partial F/\partial x,\partial F/\partial y,\partial F/\partial z]$ of the partial derivatives with respect to $x,y$ and $z$ as
\begin{equation}
\label{eq:part}
\frac{\partial F}{\partial x}\!=\! 2\!\sum\limits_{i=1}^{N}f_i\frac{\partial f_i}{\partial x}; \frac{\partial F}{\partial y}\!=\! 2\!\sum\limits_{i=1}^{N}f_i\frac{\partial f_i}{\partial y}; \frac{\partial F}{\partial z}\!=\! 2\!\sum\limits_{i=1}^{N}f_i\frac{\partial f_i}{\partial z},
\end{equation}
and express it in matrix form as $\vv{g} = 2 J^T \vv{f}$, where $\vv{f}=\{f_i, \forall i\in\mathcal{I}\}$ and $J$ is the Jacobian matrix. Now applying the Newton iteration\footnote{We adopt the Newton iteration due to its high convergence speed towards the steepest decent. However, other techniques can be readily applied, e.g., the gradient descent as reported in~\cite{Naraghi-Pour2014}.} technique~\cite{ToA-based2006}, we can find the approximate solution at step $\vv{r}_{k+1} = \{x,y,z\}$ as follows
\begin{equation}
\vv{r}_{k+1} = \vv{r}_k - (J_k^TJ_k)^{-1}J_k^T\vv{f}_k.
\end{equation} 
Note that $J_k$ and $\vv{f}_k$ are evaluated at $\vv{r}_{k}$ whereas $r_{1}=\{\hat{x},\hat{y},\hat{z}\}$ denotes the initial condition.
The time complexity of this approach is driven by the convergence speed of the Newton iteration process. Indeed, a fine-tuning can be applied to retrieve the position of the target with high accuracy within affordable time. However, dealing with mobile targets may exacerbate the problem complexity thereby reducing the accuracy of the localization system. 


\subsection{A novel technique: Pseudo-Trilateration}
\label{s:pseudo-tri}

The above-mentioned technique suffers from the following major limitations: $i$) uncertainty: wireless channels are strongly affected by fading and shadowing that may alter the distance measurements and, in turn, the calculation process and $ii$) time-variability: the classical solution requires more than four anchor nodes that might be covered by different devices (space diversity) or by the same device on different locations (time diversity); in both cases the user may move and the trilateration process may report inconsistent results.

To overcome such limitations, we propose a novel localization approach that does not require a highly-dense coverage of anchors and easily handles high-mobility users: the \textbf{Pseudo-trilateration} concept. The core idea is to use a single anchor that retrieves multiple distance measurements over a time window moving through different points, namely along some anchor motion trajectory. Such measurements are properly combined to identify the set of positions covered by the target (if moving) within the considered time window. 

 \begin{figure*}[t!]
    \centering
    \subfigure[2D localization problem with static target and single solution]
    {
       \centering
        \includegraphics[clip,trim=0cm 0cm 0cm 0cm, width=0.23\textwidth]{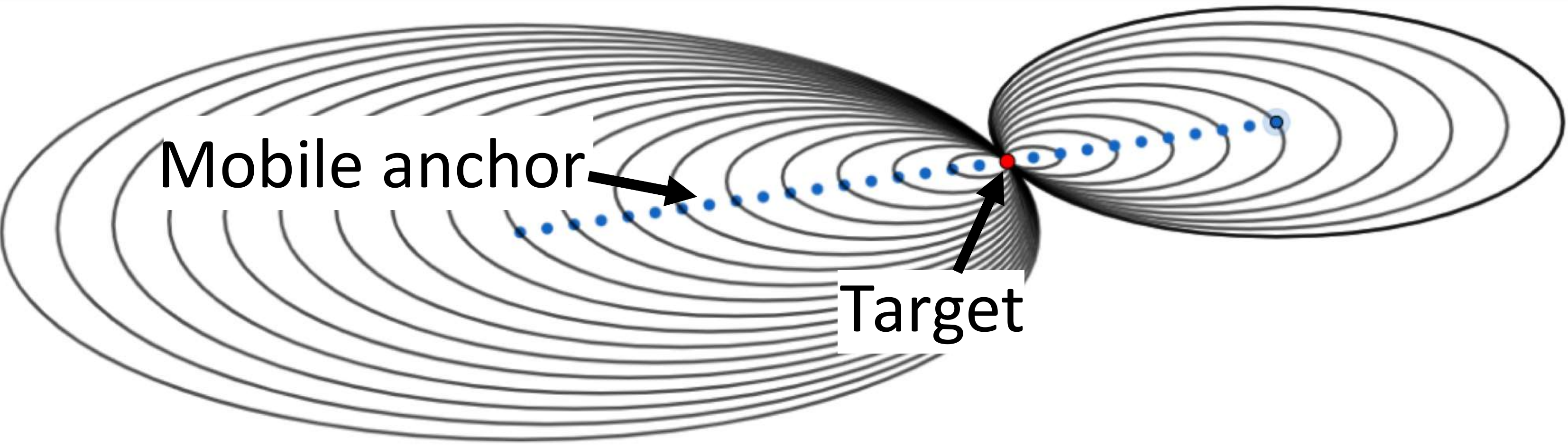}
        \label{fig:2d_1point}
    }
    \subfigure[2D localization problem with static target and double solution]
    {
        \centering
        \label{fig:2d_2points}
        \includegraphics[clip, trim=0cm 3cm 0cm 0.5cm, width=0.23\textwidth]{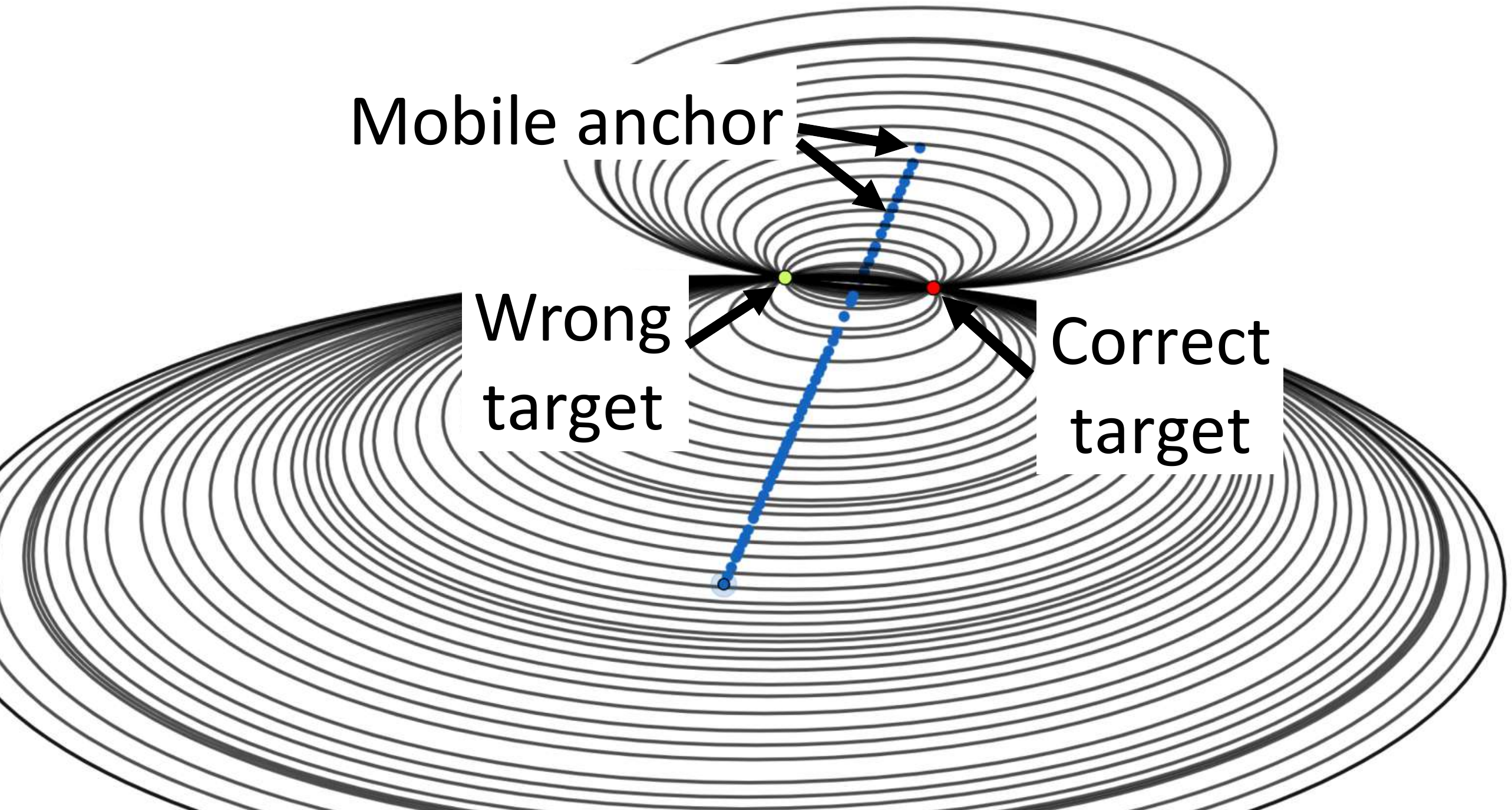}
    }    
    \subfigure[2D localization problem with mobile target and single solution]
    {
       \centering
        \includegraphics[clip, width=0.23\textwidth]{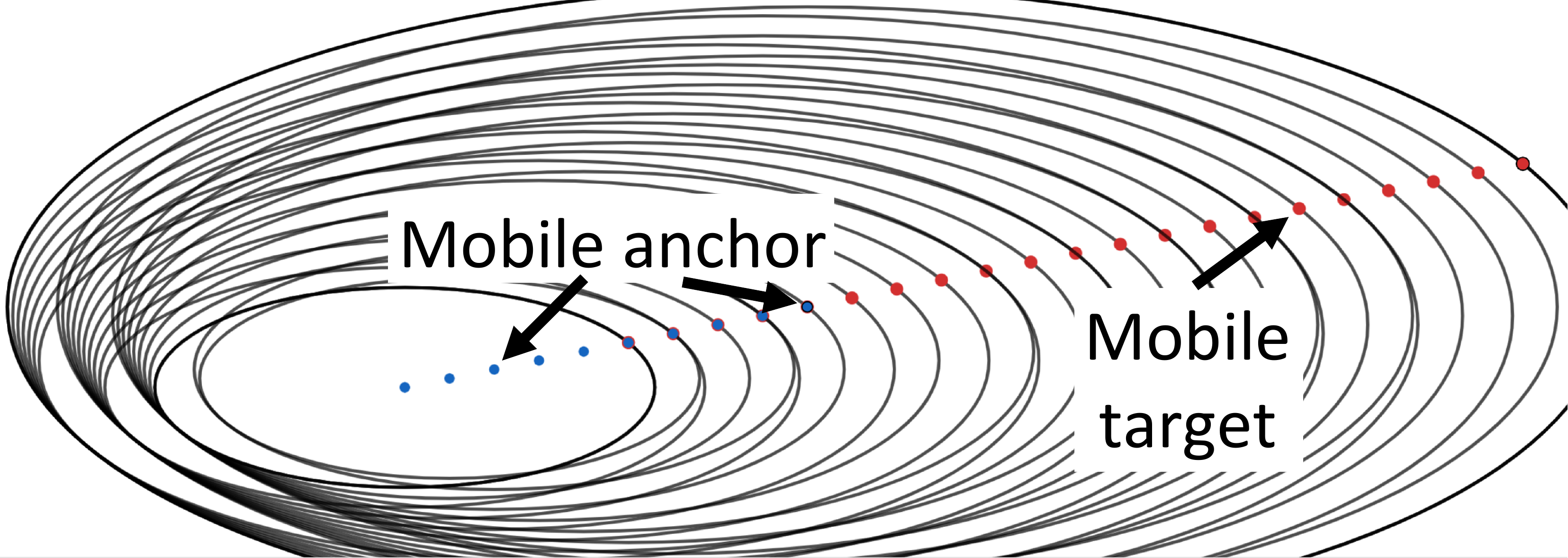}
        \label{fig:2d_1point_mob}
    }
    \subfigure[2D localization problem with mobile target and double solution]
    {
        \centering
        \label{fig:2d_2points_mob}
        \includegraphics[clip, width=0.23\textwidth]{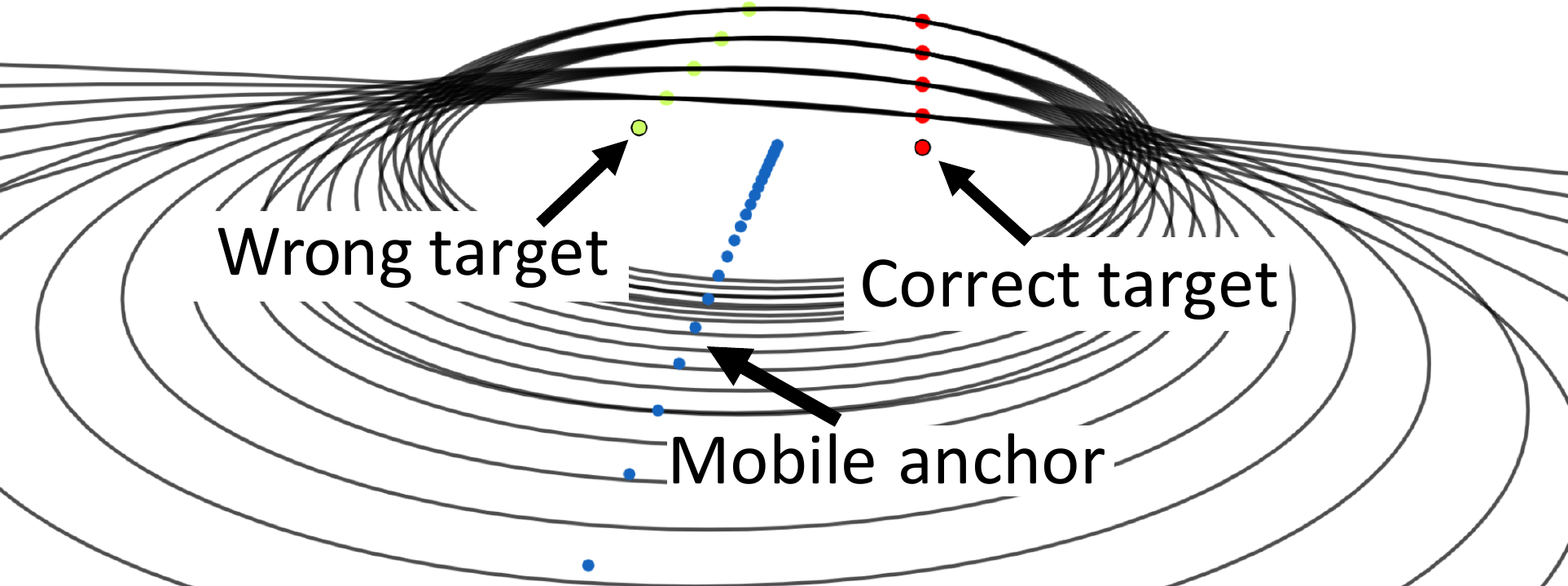}
    }    
    \label{fig:geometric}
\end{figure*}
 \begin{figure}[t!]
    \centering
   \subfigure[3D localization problem with static target and infinite solutions]
   {
       \centering
       \label{fig:3d_1circ}
       \includegraphics[clip, width=0.22\textwidth]{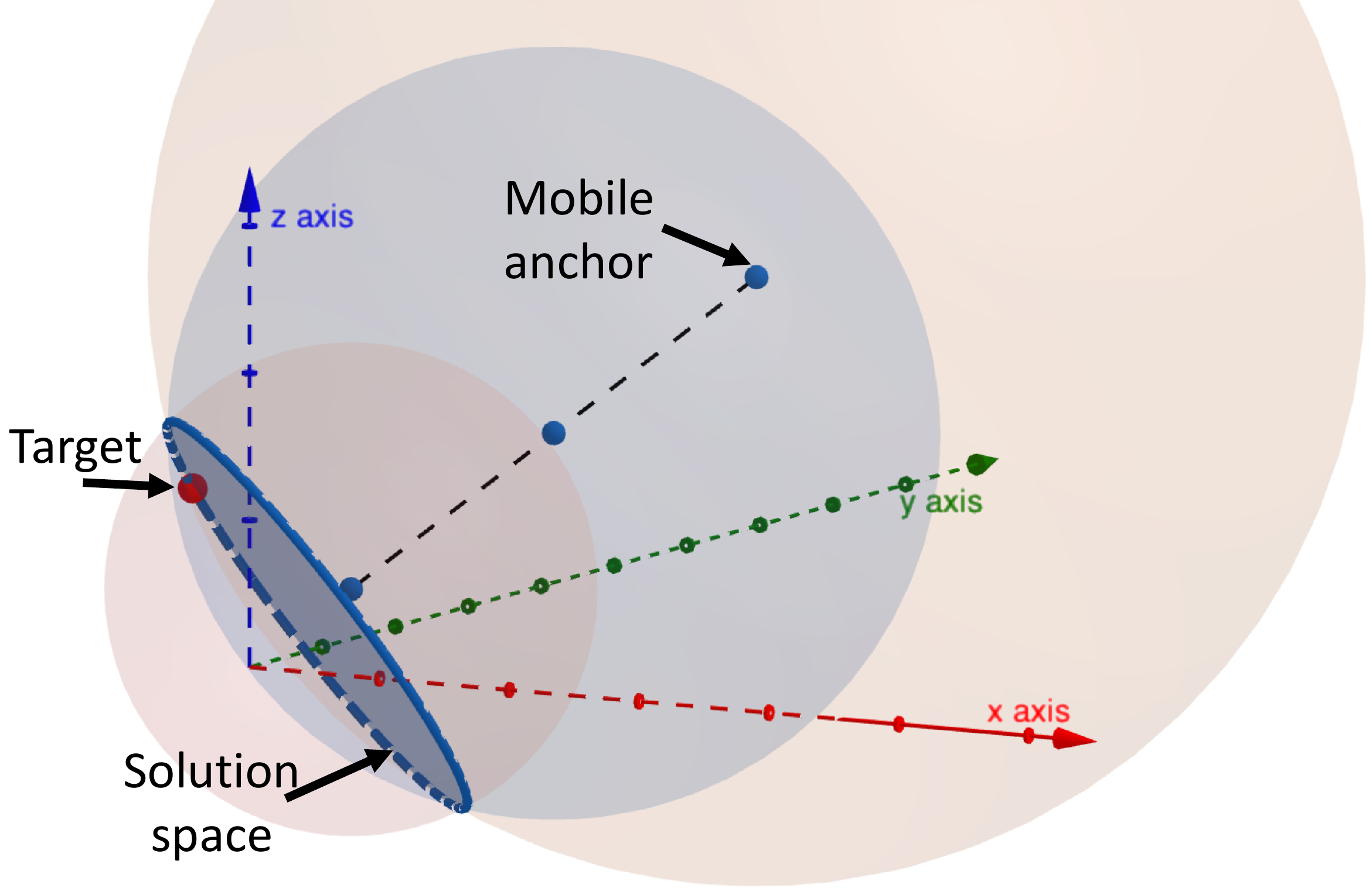}
   }
 \subfigure[3D localization problem with mobile target and infinite solutions]
    {
        \centering
        \label{fig:3d_1circ_mob}
        \includegraphics[clip, width=0.22\textwidth]{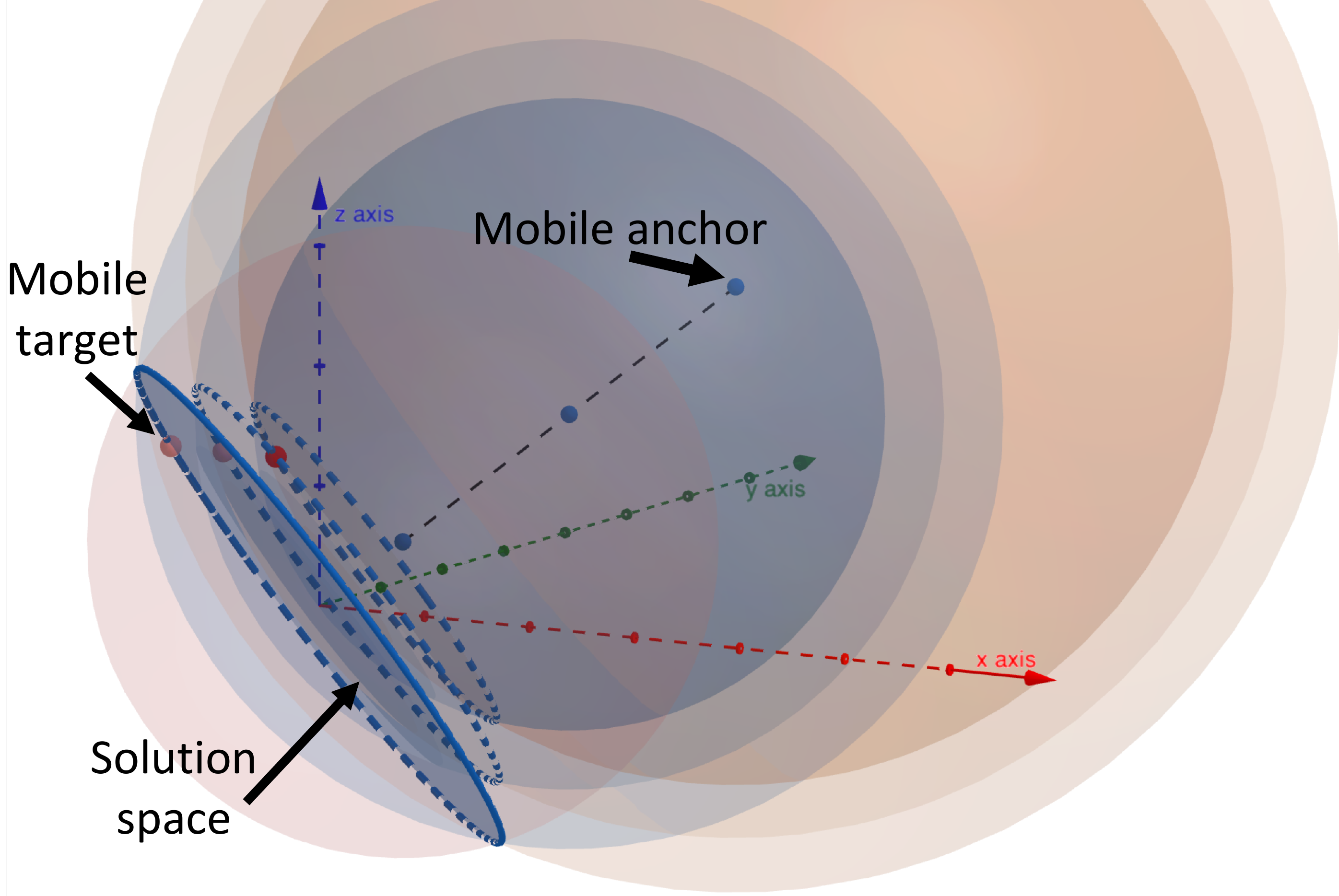}
    }
    \caption{Geometrical considerations on the localization}
    \label{fig:geometric2}
\end{figure}

Let us denote the position of our target as a vector of coordinates $\vv{x}^{(n)}=\{x_t^{(n)},y_t^{(n)},z_t^{(n)}\}$ at time $n$, where $n\in\mathcal{N}$ and $N=|\mathcal{N}|$ is the length of the time window. The motion trajectory of the anchor point is defined as the set of positions $ \vv{x}_d^{(n)}=\{x_d^{(n)},y_d^{(n)},z_d^{(n)}\}$ at each time $n\in\mathcal{N}$. 
We can rewrite Eqs.~\eqref{eq:optfunction} and \eqref{eq:function} with the ensuing equation\footnote{To avoid notation clutter, hereafter we refer to the error function  $f^{(n)}\!\!\left(\vv{x}^{(n)}\right)=f^{(n)}\!\!\left(x_d^{(n)},y_d^{(n)},z_d^{(n)}\right)$ as $f^{(n)}$.}:
\begin{equation}
f^{(n)}\!\!=\!\sqrt{\!\left(\! x_t^{(n)}\!\!-\!\!x_d^{(n)}\!\right)^2\!\!\!\!+\!\!\left(\! y_t^{(n)}\!\!-\!\!y_d^{(n)}\!\right)^2\!\!\!\!+\!\!\left(\! z_t^{(n)}\!\!-\!\!z_d^{(n)}\!\right)^2}\!\!\!-\!\hat{\gamma}^{(n)}.
\end{equation}
As per the classical trilateration technique, the objective is to minimize the function $f^{(n)}\left(x_d^{(n)},y_d^{(n)},z_d^{(n)}\right)$ that returns the error between the calculated distance $\hat{\gamma}^{(n)}$ and the real distance among the anchor point and the target every time $n$. However, in this case, only one equation is expressed per time $n$ with the objective to find the solution that provides an error equal to $0$. Unfortunately, this problem admits multiple solutions for each time $n$ as they depict the coordinates of the points that geometrically lie on a sphere with radius equal to $\hat{\gamma}^{(n)}$ centered on $\vv{x}_d^{(n)}$. 

To overcome this problem, we can minimize the sum of the Euclidean distances---namely L2-norm---between every two consecutive solutions, each of which is obtained at time $n\in\mathcal{N}$. Thus, the overall path length is minimized (\cite{humanBehaviors}).
We can formulate it with a convex optimization model:
\begin{problem}[Pseudo-trilateration]\label{problem:pseudotrilateration}
{ \begin{align}
   & \underset{\vv{x}_t^{(n)}\in\mathbb{R}^3}{\min} && \sum\limits_{n=2}^N \left\Vert \left(\vv{x}_t^{(n)} - \vv{x}_t^{(n-1)}\right)\right\Vert^2 & \nonumber\\
   & \textup{s.t.} && f^{(n)} = 0, \quad\forall n\in\mathcal{N}; \label{cst:iter} \\
   &  && \vv{x}_t^{(n)}\in\mathbb{R}^3, \quad\forall n\in\mathcal{N}. \nonumber
\end{align}}
\end{problem}
Problem~\ref{problem:pseudotrilateration} reveals its sublinearity and can be solved with commercial tool for convex optimization. We can then formulate the following theorem.
\begin{theorem}
Problem~\ref{problem:pseudotrilateration} is NP-Hard.
\end{theorem}
\begin{sketch}
The proof goes by reduction. Let us consider as input $\hat{\gamma}^{(n)}=1,\forall n\in\mathcal{N}$. Now we can find an infinite (with $N$ multiplicity) number of subsets $\vv{x}_t^{(n)}$ that satisfy the first constraint of Problem~\ref{problem:pseudotrilateration}. The new problem is to find a subset of such values with the minimum sum. This can be easily (within polynomial time) mapped onto a subset-sum problem that is known to be NP-Complete. The NP-completeness property implies that Problem~\ref{problem:pseudotrilateration} belongs to the NP class as well as to NP-hardness class. \hfill \IEEEQEDopen
\end{sketch}

With the above theorem, we prove that the solution of Problem~\ref{problem:pseudotrilateration} cannot be found within affordable time: while a low-complex heuristics can be applied to solve it, its optimality cannot be verified within polynomial time. The heuristics only solves the problem locally for each $n\in\mathcal{N}$ by finding the coordinates of the next point $\vv{x}_t^{(n)}$ with the shortest distance from the previous one $\vv{x}_t^{(n-1)}$. This results in a suboptimal global solution.

The solution of Problem~\ref{problem:pseudotrilateration} provides a set of feasible positions covered by the target within the time window $\mathcal{N}$. Note that the anchor trajectory $\mathcal{S}=\{\vv{x}_d^{(1)},\vv{x}^{(2)}_d,\ldots,\vv{x}_d^{(N)}\}$ may affect the set of solutions of the problem and, in some cases, it might result in a double optimal solution. 

For this discussion we assume full knowledge of the distance values between the anchor and the target. We discuss the 2-dimensional case, as it makes our problem tractable, however it may be extended to the 3-dimensional case by applying the same ideas. 
Let us consider a static target user, i.e., $\vv{x}_t^{(n)}=\vv{x}_t^{(m)},\forall n\neq m$ as well as a linear motion trajectory of the anchor, i.e., $y_d^{(n)}=rx_d^{(n)}+q$ with $x_d^{(n)}\neq x_d^{(n+1)}$. An example is provided in Fig~\ref{fig:2d_2points}, where two optimal solutions are depicted.
Thus, we can formulate the following lemma:
\begin{lemma}
\label{lem:bas}
Considering a static target and a linear motion trajectory of the anchor, Problem~\ref{problem:pseudotrilateration} always admits two distinct solutions. However, if the static target lies on the motion trajectory of the anchor, Problem~\ref{problem:pseudotrilateration} admits one single solution with double multiplicity.
\end{lemma}
\begin{proof}
Let us consider two different positions of the anchor $\vv{x}_d^{(n)}$ at any time $n,m$ where $n\neq m$. Using Eq.~\eqref{eq:multi-basis}, we can derive the following set of equations
\begin{align}
(x\!\!-\!\!x_d^{(n)}\!)^2\!\!\!+\!(y\!-\!(rx_d^{(n)}\!\!\!+\!\!q))^2\!\!-\!(x_t\!-\!x_d^{(n)}\!)^2\!\!-\!(x_t\!\!-\!(rx_d^{(n)}\!\!\!+\!\!q))^2\!=\!0 \nonumber \\
\label{eq:doublert} (x\!-\!x_d^{(m)}\!)^2\!\!\!+\!(y\!-\!(rx_d^{(m)}\!\!\!+\!\!q))^2\!\!\!-\!(x_t\!\!-\!\!x_d^{(m)}\!)^2\!\!\!-\!(x_t\!\!-\!(rx_d^{(m)}\!\!\!+\!\!q))^2\!\!=\!\!0
\end{align}
that provides a double root $x=[x_t\,\, -\!\frac{2qr + r^2x_t - 2ry_t - x_t}{r^2 + 1} ]$ and $y=[y_t\,\, \frac{2q + r^2y_t + 2rx_t - y_t}{r^2 + 1}]$. Clearly, when the target position lies on the motion trajectory, i.e., $y_t\!\! =\! rx_t\!+\!q$, the double root results in $x\!=\![\!x_t\,\, y_t\!]$ and $y\!=\!\![\!y_t\,\, x_t\!]$. This proves the lemma.
\end{proof}
Graphically, we show the double solution (intersection point) of Eqs.~\eqref{eq:doublert} in Figs.~\ref{fig:2d_1point}-\ref{fig:2d_2points} wherein each circumference represents one equation.

Now let us consider a nomadic target. From Lemma~\ref{lem:bas}, we can claim the following theorem.
\begin{theorem}
\label{th:close}
Problem~\ref{problem:pseudotrilateration} admits a double solution iff the motion trajectory $\mathcal{S}=\{\vv{x}_d^{(1)},\vv{x}_d^{(2)},\cdots,\vv{x}_d^{(N)}\}$ of the anchor is linear. 
\end{theorem}
\begin{proof}
Similarly to the proof of Lemma~\ref{lem:bas}, we write and solve Eqs.~\eqref{eq:doublert} for each position of the target. The solutions (single or double) of each system of equations exhibit the symmetry property with respect to the anchor trajectory. If the anchor is moving along the same direction of the target, the solution of each system is single with double multiplicity, as previously proved.
\end{proof}
A graphical illustration is provided in Figs.~\ref{fig:2d_1point_mob}-\ref{fig:2d_2points_mob}. As expected, when the anchor (blue dot) is moving along the direction of the target (red dot), the solution of Problem~\ref{problem:pseudotrilateration} reveals the exact position of the target (with double multiplicity). Conversely, when the motion trajectory is different from the direction of the target but still linear, the optimization problem results in two distinct optimal solutions. 

Analogously, in 3-dimensional space the motion trajectory affects the solution of Problem~\ref{problem:pseudotrilateration}. In this case, when the anchor direction is linear and the target position is static, there are infinite optimal solutions that geometrically lie on the circumference centered on the motion trajectory line, orthogonal to it and passing through the target position, as depicted in Fig.~\ref{fig:3d_1circ}. When the target moves, all solutions lie on analogously defined circumferences 
as depicted in Fig.~\ref{fig:3d_1circ_mob}. 

As a result, to avoid ambiguity in the solution of Problem~\ref{problem:pseudotrilateration}, based on Theorem~\ref{th:close} the anchor trajectory must change its direction within a finite time. Nevertheless, this may cause the mobile anchor to get far away from the target reducing the receive signal-to-noise ratio and so the accuracy. Therefore, we rely on \emph{a close trajectory of the mobile anchor
}. For the sake of simplicity, we model the anchor trajectory as \textit{circular}.

While considering shortest paths in our calculation may appear reasonable, in this work we go one step beyond by proposing in the next section an AI-based localization system that applies the pseudo-trilateration method, automatically learns the common motion behaviors and calculates the positions with high accuracy.

\section{\mbox{AI-Based Enhanced localization}}
\label{s:aicnn}
The pseudo-trilateration method is based on the usage of a single anchor node, e.g. a UAV flying along a predefined motion trajectory locating a target user under its cellular coverage, as detailed in Section~\ref{s:implementation}. For the sake of simplicity, we assume that the UAV keeps flying on a certain fixed altitude\footnote{While this assumption complies with national regulations, it may be simply relaxed by adding the UAV altitude to the retrieved information.}. 

We develop a machine learning framework to solve the positioning problem, that is to infer the position of the target, dubbed as UE according to the LTE nomenclature, given the raw ranging measurement series. In the following, we provide some insights on the measurement data preprocessing, the Convolutional Neural Network (CNN) design as well as on its training dataset. 

\subsection{Data preprocessing}
\label{s:data_preprocessing}
Let us consider any close trajectory along which the UAV may fly during the ToF measurements collection. For the sake of tractability, we assume that the UAV takes measurements in a subset of equally-spaced points on its trajectory, namely measurement spots. Moreover, we assume that all measurements are taken at the same time instant, meaning that there is no time delay between any measurement in the same spot\footnote{This assumption makes our analysis tractable. However, in practice the ToF measurements are performed within few milliseconds, as shown in Section~\ref{s:results}.}. Let us denote as $N$ the number of measurement spots and $L$ the number of measurements taken for each of them. The set of measurements taken by the UAV during a single revolution is $\hat{\gamma}^{(n)}_l$, $\forall n \in \mathcal{N}, l \in \mathcal{L}$, with  $N = |\mathcal{N}|$ and $L = |\mathcal{L}|$. For each revolution, we arrange the measurement series into a 2D data structure where the time and spatial relationships among subsequent measurements may be better represented. In particular, we build a matrix $\Gamma_m$, $\forall m \in \mathcal{M}$, of size $N \times L$, where $\mathcal{M}$ is the set of all close trajectories. To take into account the UAV trajectory, we concatenate each $\Gamma_m$ with a matrix $D_m$ that contains the 3D coordinates of all UAV measurements spots, as depicted in the following. We name the above mentioned matrix $\Phi_m$ and express it as follows 
\begin{align}
\label{eq:phi}
\Phi_m = \begin{bmatrix}
\hat{\gamma}^{(1)}_1 & \dots  & \hat{\gamma}^{(1)}_L & x_1 & y_1 & z_1\\
\hat{\gamma}^{(2)}_1 & \dots  & \hat{\gamma}^{(2)}_L & x_2 & y_2 & z_2\\
\vdots & \ddots & \vdots & \vdots & \vdots & \vdots\\
\bundermat{$\Gamma_m$}{\hat{\gamma}^{(N)}_1 & \dots  & \hat{\gamma}^{(N)}_L} & \bundermat{$D_m$}{x_{N} & y_{N} & z_{N}} 
\end{bmatrix}.
\end{align} 
\vspace{4mm}

\noindent To understand the rationale behind our neural network design, let us assume a scenario wherein the target UE is partially covered by a fixed obstacle, e.g. rubble. As shown in Fig.~\ref{fig:data_sample}, the heatmap of the corresponding $\Gamma_m$ presents a stripe pattern due to the excess error at the measurements spots where the wireless channel suffers from the shadowing effect induced by the obstacle.  
\begin{figure}[t!]
    \centering
    \subfigure
    {
       \centering
        \includegraphics[clip, width=0.24\textwidth]{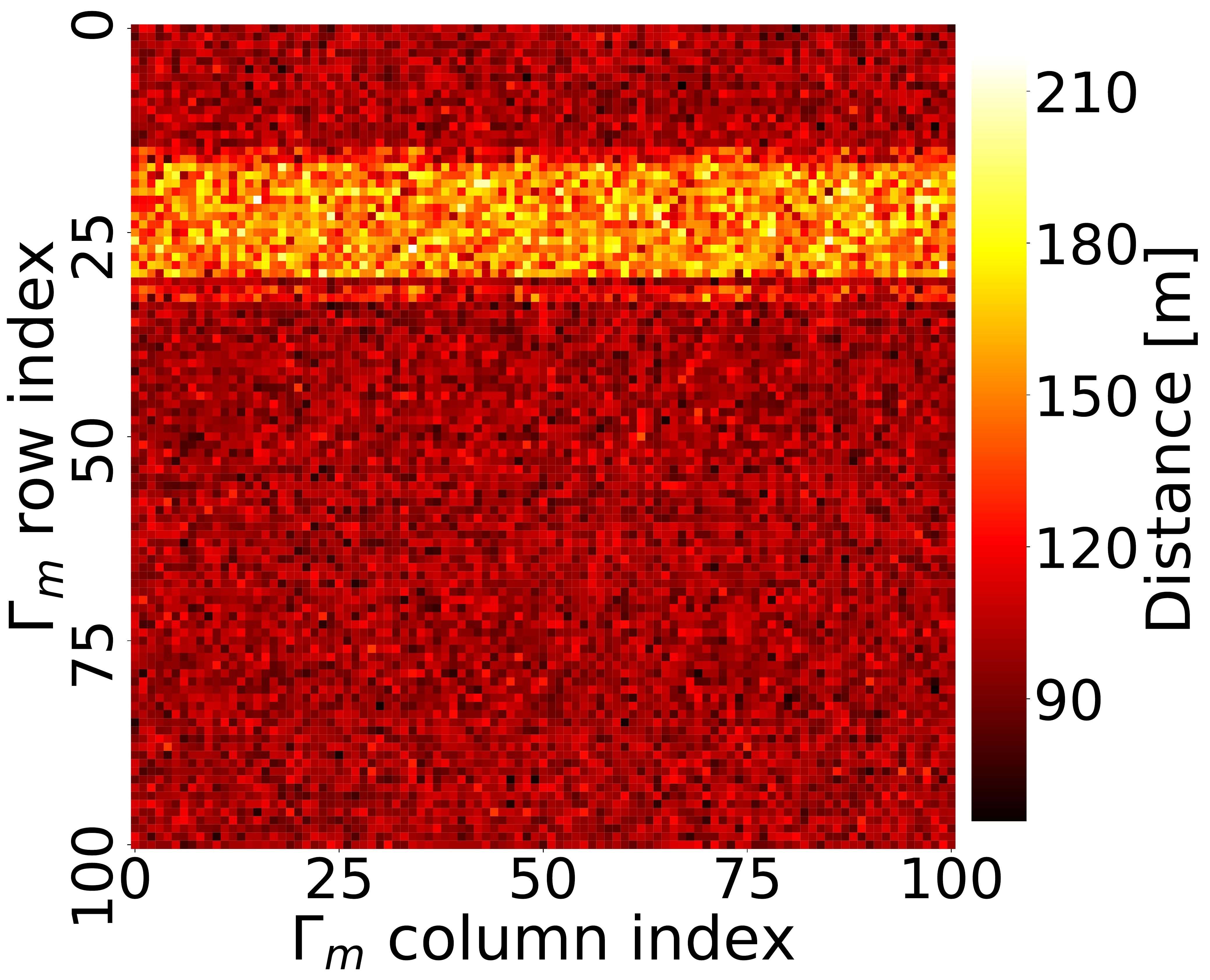}
    }
    \subfigure
    {
       \centering
        \includegraphics[clip, width=0.19\textwidth]{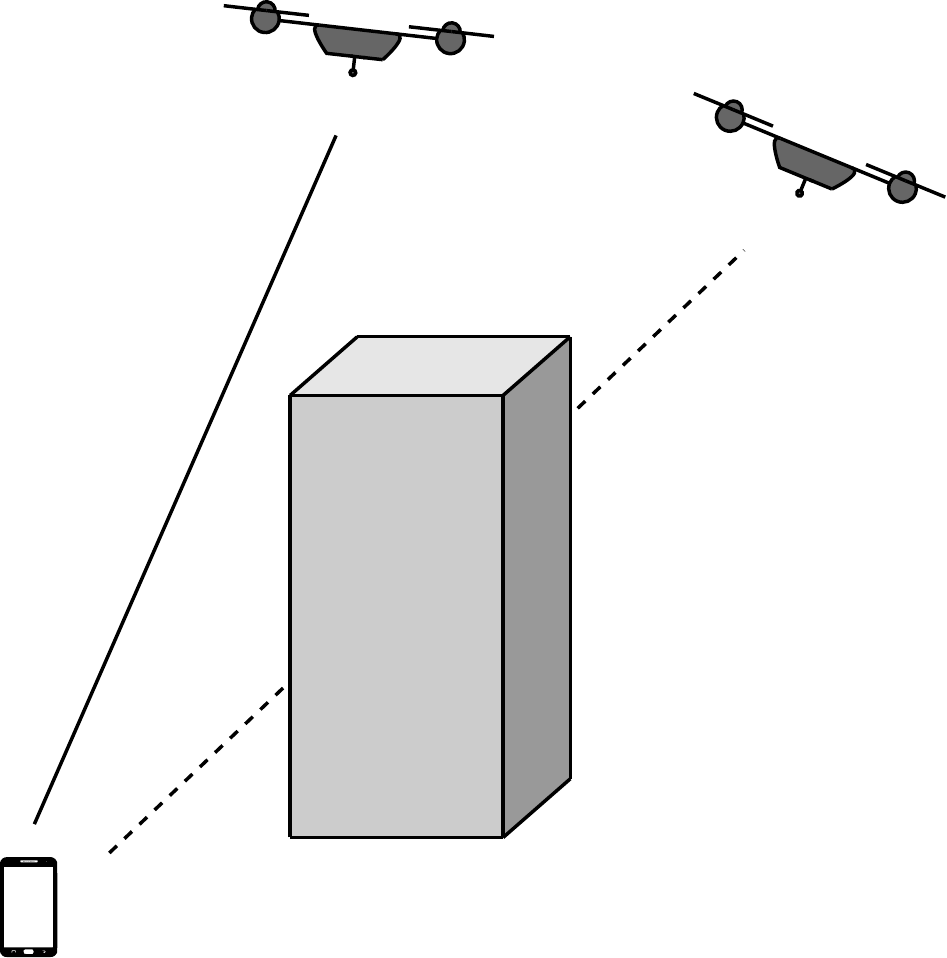}
    }
    \caption{Heatmap of matrix $\Gamma_m$ of distance measurements in a single-obstacle scenario.}
    \label{fig:data_sample}
\end{figure}
Different propagation environments lead to different patterns in $D_m$ thereby suggesting to process each corresponding $\Phi_m$ as a single-channel picture. In such regard, we employ a 2D CNN whose usage is well established in the image processing field, e.g. for target recognition~\cite{Zhao2018}.
\subsection{2D CNN design and training}
\label{s:cnn}
2D CNNs are a class of feed-forward neural networks that make use of convolutional layers to extract information from 2D input data. Generally, they consist of a bunch of modules that include one convolutional and one pooling (or subsampling) layers, and they may be repeated to build a deeper model. In addition, some fully connected layers are stacked onto the last module to provide the final output of the network. Although CNNs are very much used to perform classification tasks, we employ them to solve our localization problem. We train the network to infer the user trajectory, given $5000$ input data samples. In particular, for each $\Phi_m$, we train the network to regress $U_m$, which denotes the matrix of the 2D user coordinates $x_t^{(n)}$, $\forall n \in \mathcal{N}$ corresponding to each UAV measurement spot. We denote as $\hat{U}_m$ the matrix of the regressed user coordinates $\hat{x}_t^{(n)}$, $\forall n \in \mathcal{N}$. 
 
\begin{figure}[t!]
      \centering
      \includegraphics[clip, width=0.49\textwidth ]{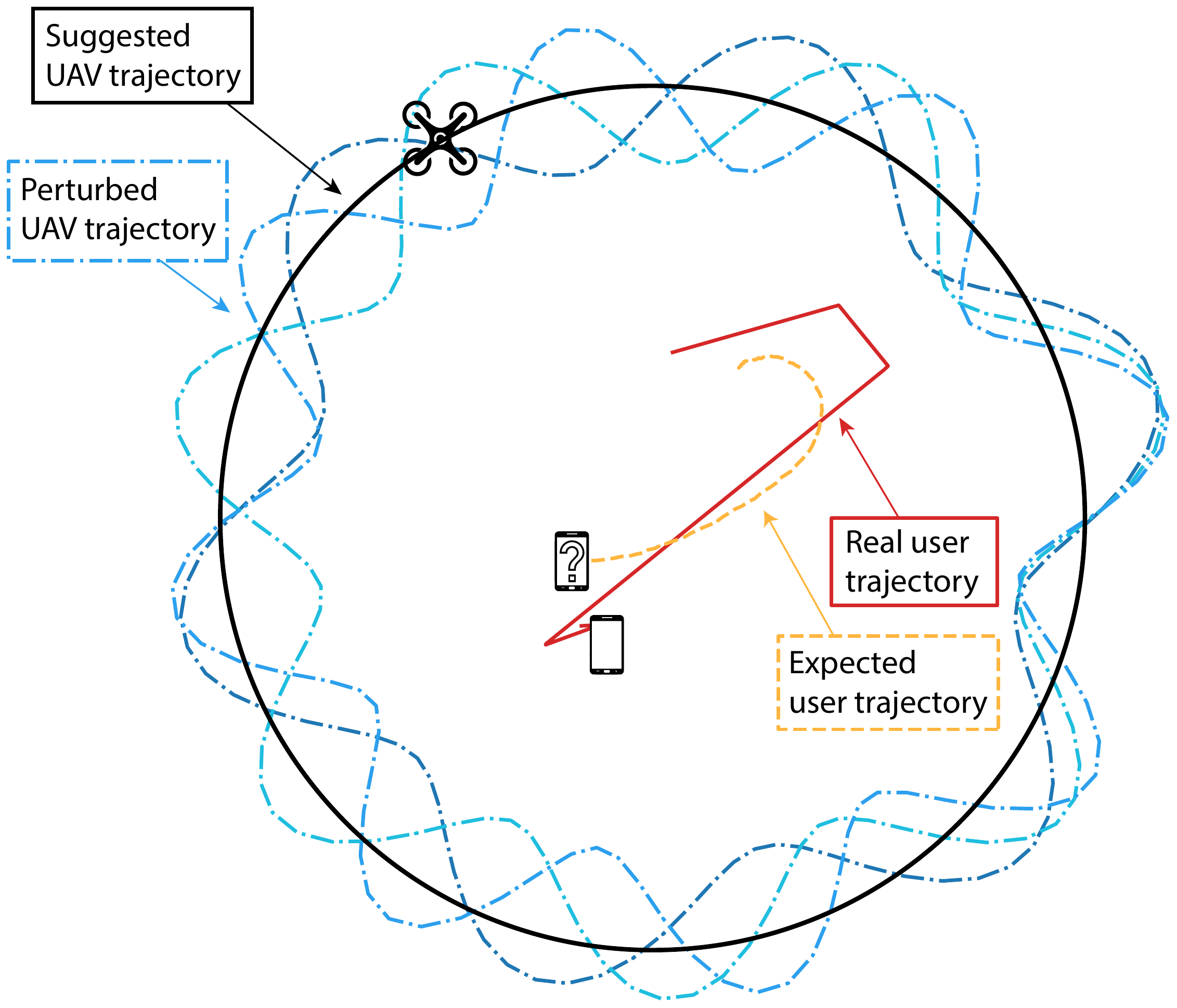}
      \caption{Real example of localization error on simulated data applying the pseudo-trilateration mechanism considering UAV circular perturbed trajectory.}
      \label{fig:trajectories}
\end{figure}

Our CNN design involves two modules and three fully connected layers. During the training phase, we aim at offering the network a broad range of labeled samples. To do so, we simulate several UAV and user trajectories, and generate a synthetic dataset as described in Section~\ref{s:data_preprocessing}. In Section~\ref{s:pseudo-tri}, we have shown that a close line satisfies Theorem~\ref{th:close}. Indeed, to take into account any possible perturbation of the imposed UAV trajectory, we simulate UAV trajectories considering sinusoidal characteristics~\cite{unmannedAircraft_book} according to the following equation:
\begin{equation}
\label{eq:trajectory}
\begin{split}
\vv{x}_d^{(n)} = \{(\rho + asin(2 \pi n/ N)) \, cos(2 \pi/N) + x_c, \\
(\rho + asin(2 \pi n/ N)) \, sin(2 \pi / N) + y_c, h\},
\end{split}
\end{equation}
$\forall n \in \mathcal{N}$, where $\{x_c,y_c,h\}$ and $\rho$ denote the coordinates of its center and radius, respectively, and $a \in \mathbb{R}$ is a parameter such that $0<a<\rho$. Some examples of such trajectories are shown in Fig.~\ref{fig:trajectories}. Besides, to simulate realistic user trajectories, we employ the so called SLAW mobility model~\cite{SLAW}.

\subsection{Encoder-Decoder LSTM design and training}
\label{s:lstm}

We set the properties of the next UAV trajectory on the base of the predicted future user behavior. In order to perform this prediction, we employ a multi-layered Long Short-Term Memory (LSTM) that takes as input the current 2D-CNN output, namely the estimated current user trajectory, and outputs its expected evolution. This problem is usually denoted as sequence-to-sequence prediction problem.

Our method is based on the so-called Encoder-Decoder LSTM Recurrent Neural Network~\cite{GoogleLSTM}. The core idea is to map the input sequence to a fixed-length vector using an LSTM, which is the Encoder, and then map the latter vector onto the target sequence, which is the Decoder LSTM. The Encoder and the Decoder LSTMs are then followed by a fully connected feed-forward layer, which represents the output layer. In this way, the network creates an internal fixed-dimensional vector representation of the input sequence and learns how to generate an output sequence of the same or different length, namely the prediction.

Specifically, we train the network with a dataset made of input-output pairs $(U_m,F_m)$, $\forall m \in \mathcal{M}$, where $F_m$ is the matrix containing the future user coordinates, that is $\vv{x}^{(f)}$, $\forall f \in \mathcal{F}$, with $F = |\mathcal{F}|$. We name $\hat{F}_m$ the matrix containing the predicted user coordinates $\hat{\vv{x}}^{(f)}$, $\forall f \in \mathcal{F}$.


\section{Dynamic UAV Relocation}
\label{s:controller}

\begin{algorithm}[t]
\centering
\begin{framed}
\begin{enumerate}[leftmargin=*]
\small
\item Initialise the offset values $\delta_x,\delta_y,\delta_\rho$ to $0$.
\item Update the UAV trajectory parameters with the following equations $x_c+\!\!=\delta_x,\,\,y_c+\!\!=\delta_y,\,\,\rho+\!\!=\delta_\rho$.
\item Calculate the spatial average $\bar{\vv{x}_t}=\{ \bar{x}_t,\bar{y}_t\}$ of the predicted user positions $\hat{\vv{x}_t}^{(f)} =\{\hat{x}_t^{(f)},\hat{y}_t^{(f)}\}$ with the following
\[
\bar{x}_t = \frac{1}{F} \sum\limits_{f=1}^{F} \hat{x}_t^{(f)}, \quad\bar{y}_t = \frac{1}{F} \sum\limits_{f=1}^{F} \hat{y}_t^{(f)}.
\]
\item Calculate the average predicted speed of the user as follows
\[
\bar{v}_t = \frac{1}{F-1}\sum\limits_{f=2}^F ||\hat{\vv{x}}_t^{(f)}-\hat{\vv{x}}_t^{(f-1)}||^2.
\]
\item Set $\delta_x\! =\! x_c-\bar{x}_t \quad \delta_y\! =\! y_c-\bar{y}_t \quad \delta_\rho\!=\!\frac{\bar{v}_t}{\bar{v}_d}\max\limits_{f\in\mathcal{F}} ||\vv{x}_t^{(f)}-\bar{\vv{x}}_t||^2$.
\item Go to Step 2.
\end{enumerate}
\end{framed}
\caption{UAV Relocation procedure}
\label{algo:reloc}
\end{algorithm}

Predicted user positions are used to improve the localization process. After a complete revolution, \name{} automatically adjusts the UAV position to get closer to the user so as to retrieve more accurate distance measurements. To make the adjustment process simple, we assume that the UAV is instructed to change the radius $\rho$ of the circular motion trajectory by an offset, namely $\delta_{\rho}$ and the position of its center $\vv{x}_c$ by a space offset, namely $\delta_x,\delta_y$ (and $\delta_z$ in case of $3-$dimensional scenarios). The pseudo-code is listed in Algorithm~\ref{algo:reloc}. Predicted user positions $\hat{\vv{x}}_t^{(f)}$ within the future time window $\mathcal{F}$ are retrieved from the Encoder-Decoder LSTM block, as explained in Section~\ref{s:lstm}. Such coordinates are spatially averaged to get the next center value of the UAV trajectory $x_c,y_c$. The maximum distance from that depicts the radius of the UAV trajectory. Additionally, we compute a confidence value that depends on the user average speed $\bar{v}_t$ and the UAV speed $\bar{v}_d$: if the user is moving faster our algorithm adds a safety margin to the radius offset $\delta_\rho$ to keep the user close to the UAV trajectory coverage.

\begin{figure}[t!]
      \centering
      \includegraphics[clip, width=0.9\linewidth ]{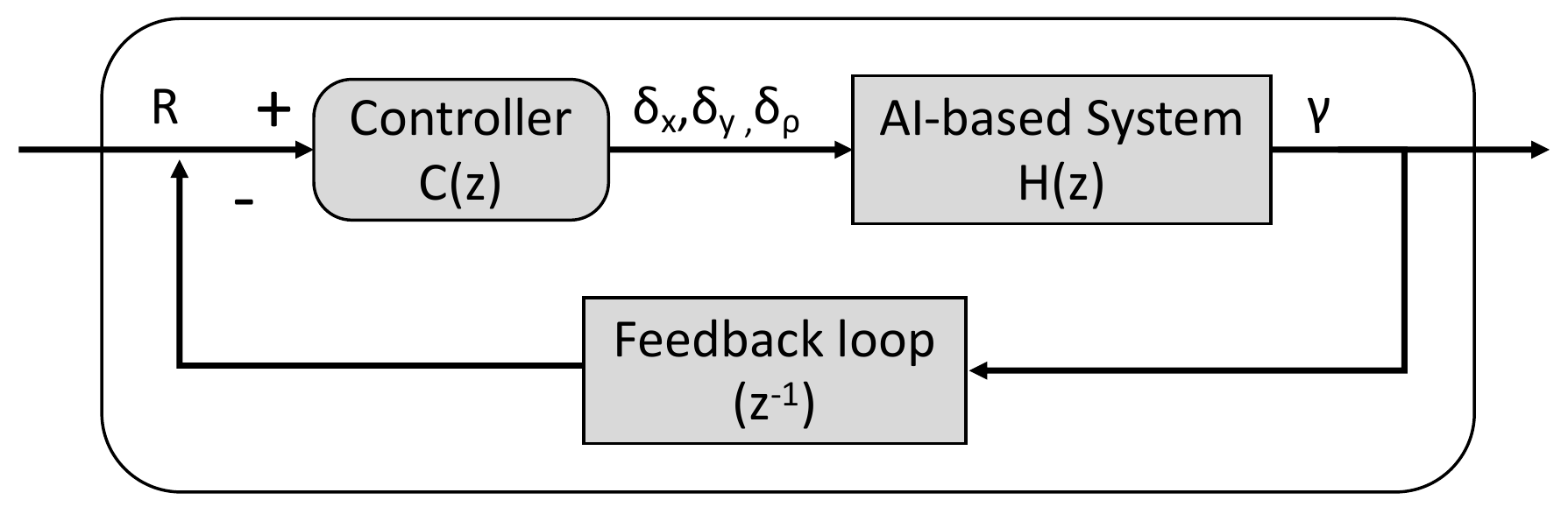}
      \caption{Building blocks of the controller designed to automatically adjust the UAV position.}
      \label{fig:controller}
\end{figure}
Recalling from control theory, our controller aims at reducing the difference between the reference signal $R(t)$ and the output signal $O(t)$. Specifically, in our case we design the output signal from the overall system as the distance between the UAV and the user, i.e., $\gamma^{(n)}$ whereas the reference is set to $0$, as shown in Fig.~\ref{fig:controller}. Assuming our system running in a discrete domain, we can write the $z$ transform of the controlled system as the following\footnote{The transfer function is linearized at (0,0).
}
\begin{align}
&H(z)\!=\!2\delta_x(x_c\!+\!(\rho\!+\!\delta_\rho\!)cos\omega\!-\!\bar{x}_t)
\!\!+\!\!2\delta_y(y_c\!+\!(\rho\!+\!\delta_\rho)sin\omega\!-\!\bar{y}_t)\nonumber\\
&+2\delta_\rho((\rho\!+\!\delta_\rho\!+\!(x_c\!+\!\delta_x\!-\!\bar{x}_t)cos\omega\!+\!(y_c\!+\!\delta_y\!-\!\bar{y}_t)sin\omega
\end{align}
whereas the PI controller\footnote{We design a Proportional-Integral (PI) controller due to its simplicity while guaranteeing zero error in the steady-state.} is defined as $C(z) = K_p + \frac{K_i}{z-1}$. Following the Ziegler-Nichols rules~\cite{franklin1990}, we can calculate $K_p=0.1$ and $K_i=0.11$ to keep our system stable.

Considering multiple UAV positions $n\in\mathcal{N}$, the controller adaptively tunes the parameters to reduce the distances $\gamma^{(n)},\forall n\in\mathcal{N}$ that in turn translates in having the UAV trajectory accurately covering the expected positions of the user $\hat{x}_t,\hat{y}_t$. Specifically, when the user moves around certain points showing a limited motion area, our UAV tries to reduce the circular trajectory area, i.e., reduces $\delta_\rho$ to focus and retrieve higher accurate measurements (in the extreme case when the user is static, the UAV is covering the minimum trajectory area). Conversely, when the user moves over a larger area the controller increases the coverage area trying to keep the user within the close trajectory. In Section~\ref{s:results}, we show the performance results in terms of system stability over time and we finally show the overall system performance (running our algorithm) when the feedback loop is active against different user speeds. 


\section{\name~Implementation} \label{s:implementation}
\begin{figure}[t!]
      \centering
      \includegraphics[clip, width=0.8\linewidth ]{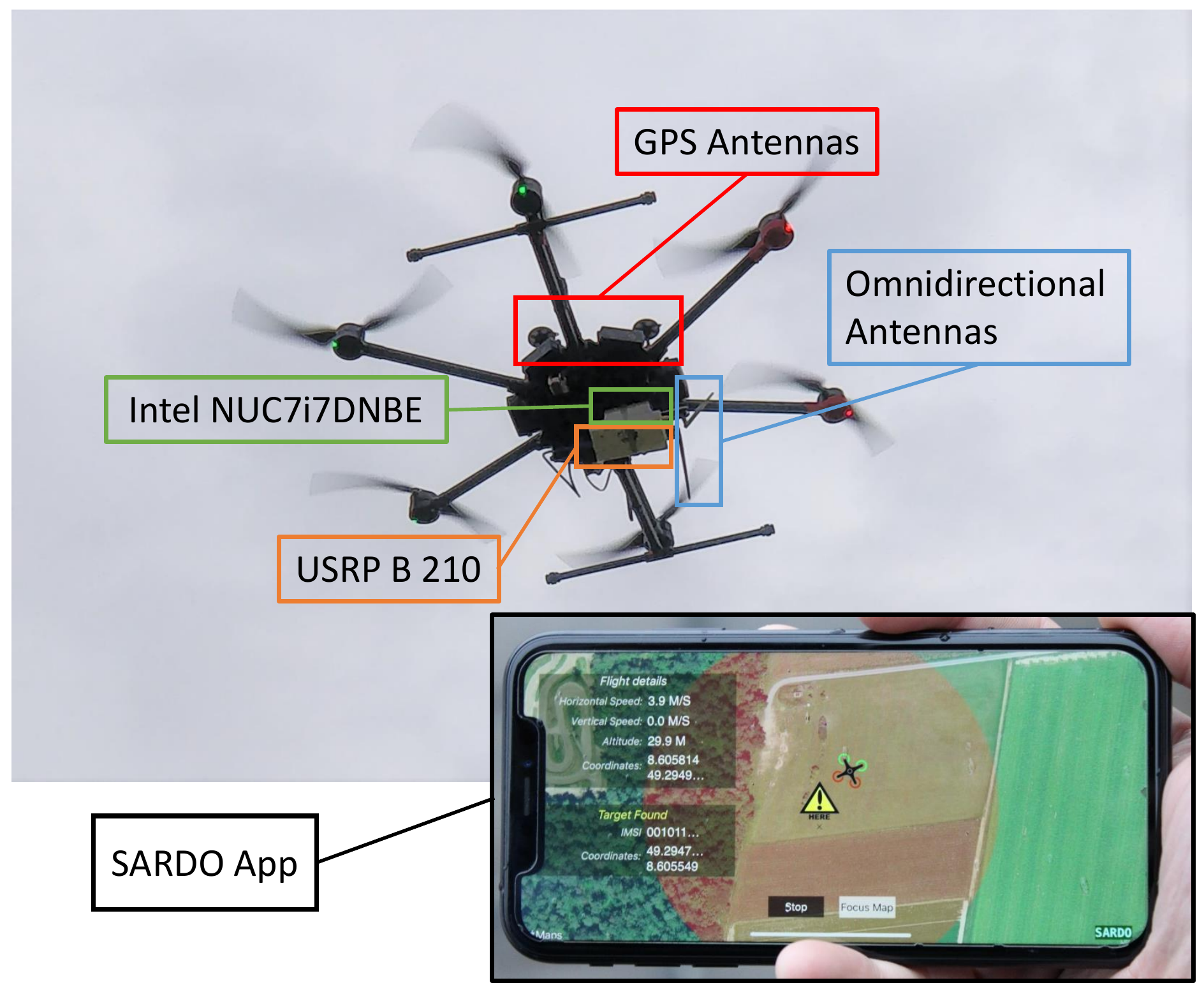}
      \caption{
      \name{} in action. Video available at \newline \url{https://www.youtube.com/watch?v=9v3NNghq3so}.}
      \label{fig:drone_app}
\end{figure}


We designed and prototyped \name{} going through an intensive engineering process resulting in the architecture summarized in Fig.~\ref{fig:architecture}. \name{} runs on off-the-shelf equipment (listed in Table \ref{tab:equipment_list}) and does not require any protocol stack change to mobile users devices. 
\begin{table}[ht]
\caption{Equipment for \name{} prototype}
\label{tab:equipment_list}
\scriptsize
\centering
\resizebox{0.4\textwidth}{!}{%
\begin{tabular}{c|c}
\textbf{Equipment} & \textbf{Model}\\  
\hline
\rowcolor[HTML]{EFEFEF}
FGPA Board  &   NI USRP B210 \\
Embedded Computer & Intel NUC7i7DNBE \\
\rowcolor[HTML]{EFEFEF}
UAV     &  DJI Matrice 600 PRO\\
\makecell{High-Gain Amplifier\\Low-Noise Amplifier} & \makecell{Mini-Circuits ZX60-V63+ \\ ZX60-33LNR-S+ \cite{amplifiers}}\\
\rowcolor[HTML]{EFEFEF}
Directional Antennas    &  $2\,\times$ $10$ dBi\\
UE &    Samsung Galaxy Tab S2 \\
\hline
\end{tabular}%
}
\end{table}

{\bf Prototype architecture.} The first building block is in charge of collecting ToF measurements from the target UE and processing them as described in Section~\ref{s:data_preprocessing}. 
This block relies on the Software Defined Radio technology building on top of srsLTE~\cite{srsLTE}, an open-source LTE-compliant software suite~\cite{srsLTEProd}---deployed on an Intel NUC Board~\cite{intelNUC} with $32$GBs of RAM and $1.9$ GHz $7$th generation CPU---that interfaces with an FPGA board, NI USRP B210~\cite{usrpb210}, equipped with omnidirectional antennas. In order to devise an all-in-one solution, we deploy the network backhaul and core domain as part of the srsLTE suite. This enables a quick and direct interaction with any single 3GPP architectural component, such as SGW, PGW, HSS or MME (hereafter described).

{\bf Testbed.} We securely set up this module of about $1$kg weight on board of an advanced DJI Matrice 600 Pro UAV~\cite{drone600pro} that is able to carry up to $6$kg payload, as shown in Fig.~\ref{fig:drone_app}. In addition, to control the UAV trajectory, we developed a control iOS application by means of the DJI Mobile SDK\footnote{The iOS application is only needed to control the UAV motion patterns while processing runtime information. However, first responders may run a standalone application that automatically triggers new UAV directions.}, which has a twofold function: $i)$ it fetches the UAV coordinates and relays them to the NUC board, where first the 2D CNN and then the Encoder-Decoder LSTM are executed to calculate the new UAV trajectory parameters according to Algorithm~\ref{algo:reloc}, $ii)$ it retrieves such new settings and calculates the next UAV positions that are set back on the UAV.
This is achieved by means of a 2.4 GHz WiFi control channel that delivers information to the Intel NUC while a proprietary DJI wireless communication interface is used by the DJI framework, i.e. by the ad-hoc iOS application, to deliver new UAV motion patterns. An anonymous online video is available at~\cite{onlinevideo}. 

To make our solution mobile infrastructure independent, as highlighted in Section~\ref{s:intro}, we design the ToF measurements processing module such that it does not require successful associations with the user equipments (UEs). 
For the sake of clarity, in the following we do not distinguish between the above mentioned module and the UAV itself.

\vspace{2mm}\subsection{Mobile Infrastructure Independence}
\name{} does not rely on the successful UE attachment thereby reducing the overall complexity of the system and improving the effectiveness of rescue operations in emergency scenarios. 

\begin{figure}[t!]
      \centering
      \includegraphics[clip, width=0.47\textwidth ]{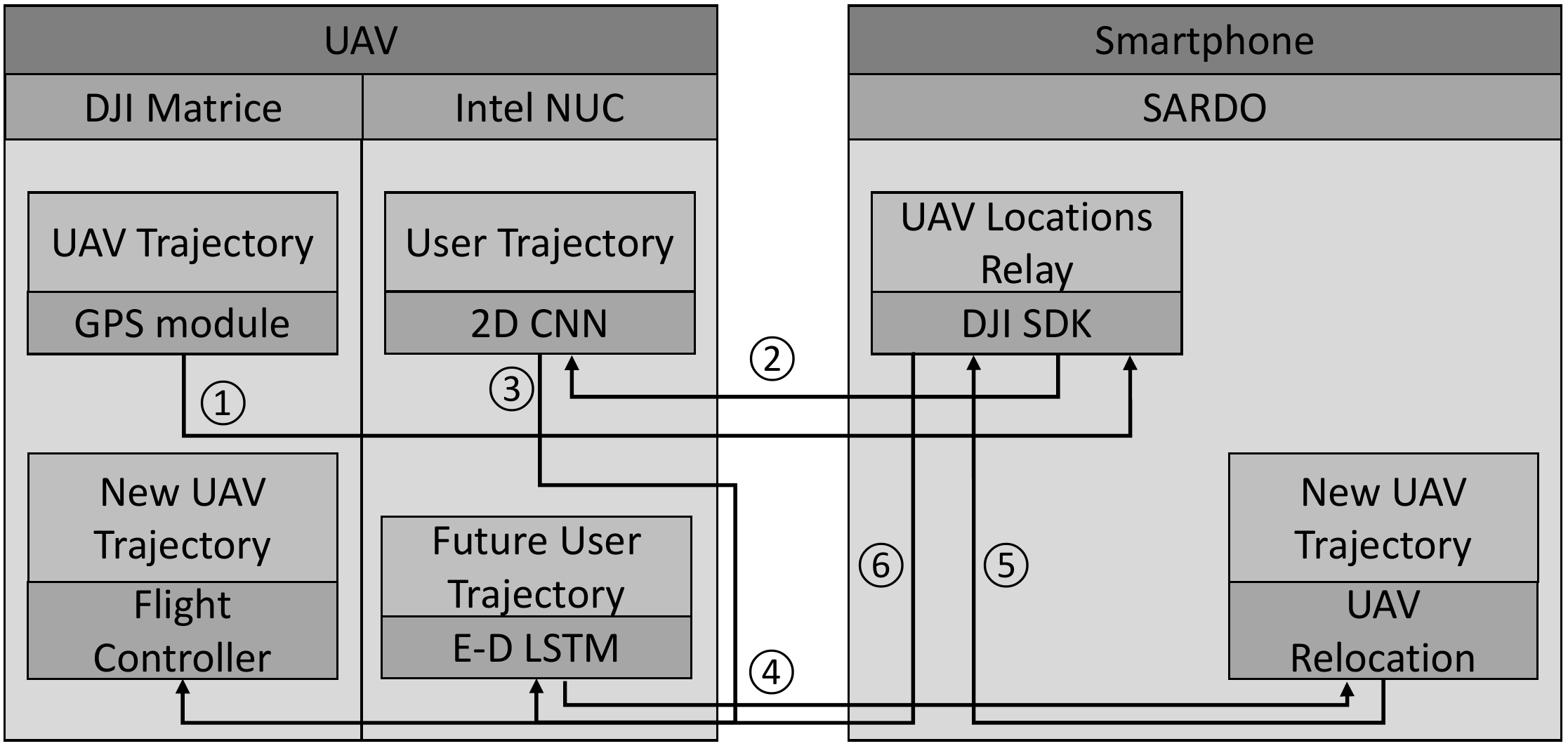}
      \caption{\name{} prototype architecture with message flow sequence numbers.}
      \label{fig:architecture}
\end{figure}

The 3rd Generation Partnership Project (3GPP) prescribes that the UE performs the Random Access Procedure (RAP) whenever it attempts to establish a connection with a base station (namely eNodeB or eNB), e.g., initial access to the network or handovers~\cite{3gpp_IDLE}. However, the completion of the RAP does not imply that the UE is attached to the eNB. Indeed, the UE establishes a Radio Resource Control (RRC) connection with the selected eNB to access the required network resources. In particular, the RRC layer is responsible for radio resource configuration and mobility management of connected UEs. In addition, RRC serves as transport protocol for Non-Access Stratum (NAS) signaling messages between a UE and its Mobile Management Entity (MME). 

We introduce specific changes to the srsLTE software but keeping our system in full compliance with 3GPP standard guidelines.
We update the NAS signaling\footnote{\name{} achieves the disclosure of the UE identity relying on a design choice of the current LTE standard. However, next cellular network generations may natively provide such information.} for the Tracking Area Update (TAU) Procedure so that the UE reveals its International Mobile Subscriber Identity (IMSI) while exchanging messages with the base station, i.e., our UAV cell. In this way, it is straightforward to identify the UE and start the localization process within a short disruption time window. Such an identity-awareness feature opens up new use cases, e.g. searching for specific missing people or locating specific targets for public safety purposes.
%
Thus, our ToF measurements processing module is built as an IMSI-catcher~\cite{IMSI_catcher}. Hereafter, we detail the minimal steps of our approach. 

\begin{figure}[t!]
      \centering
      \includegraphics[clip, width=0.47\textwidth ]{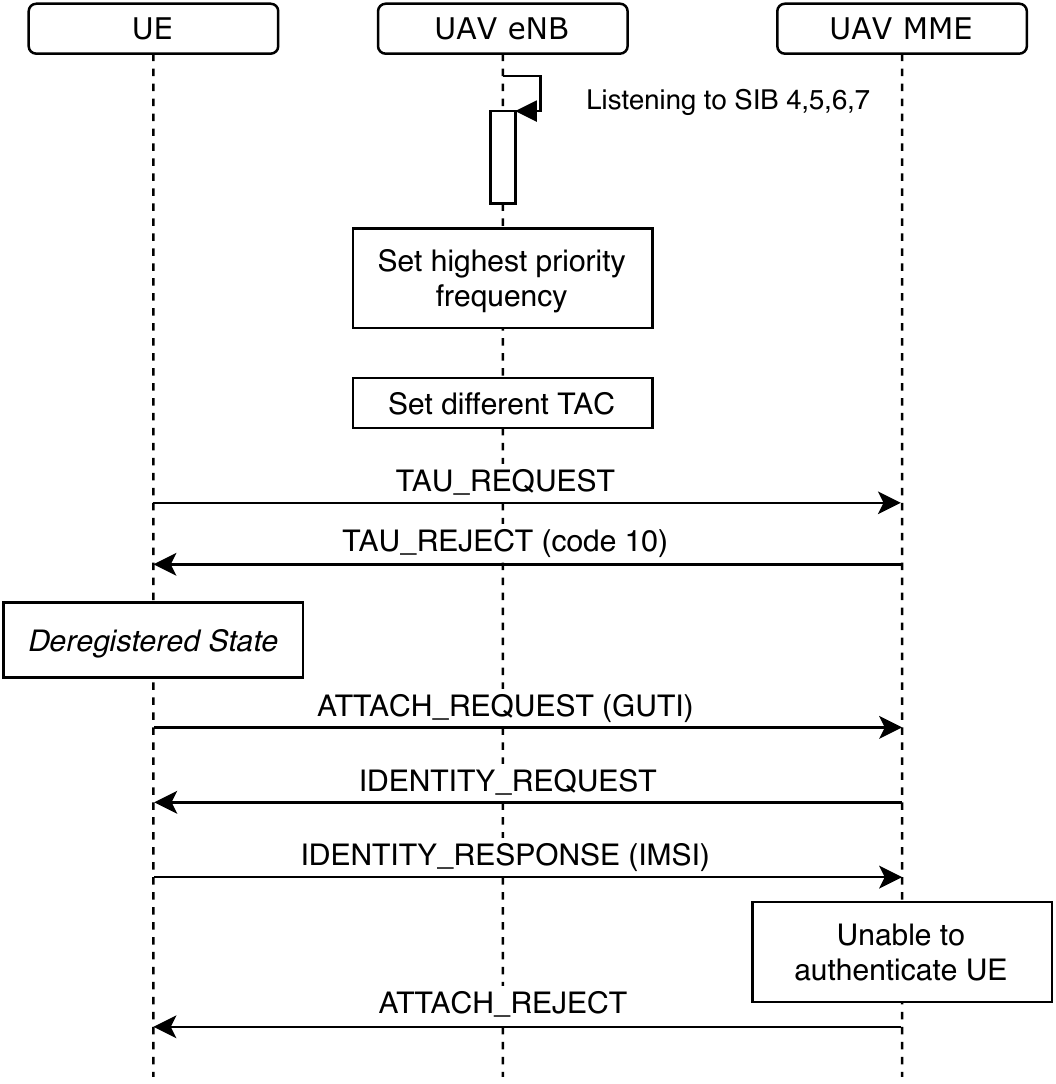}
      \caption{IMSI-Catcher message sequence chart }
      \label{fig:IMSI_catcher_chart}
\end{figure}
{\bf IMSI-catcher.} The message sequence chart is depicted in Fig.~\ref{fig:IMSI_catcher_chart}. The UAV listens to System Information Block (SIB) messages transmitted by existing ground base stations in the disaster area. SIB Type 4, 5, 6 and 7 messages carry the parameters of the Cell Reselection Procedure. Such procedure aims at moving the UE to the best cell of the selected operator. In LTE, this is accomplished by letting the UE assess all the frequencies and Radio Access Technologies (RATs) based on the priority list specified in the above mentioned SIB messages and then rank them according to the respective radio link quality. The TAU procedure is triggered as soon as the UE selects a new Tracking Area Code (TAC) different from the one it is currently camping on. Therefore, we setup our UAV on the highest-priority frequency. 

This deteriorates the received power of the serving eNB, thus enabling the inter-frequency and inter-RAT Cell Search Procedure, which eventually leads to a Cell Reselection towards our UAV.
Moreover, we set a different TAC with respect to the camping one in order to trigger a \texttt{TAU\_REQUEST} from the UE. To this request, our UAV MME responds with a \texttt{TAU\_REJECT} with cause \#10: \texttt{Implicitly detached}. In this way, the UAV MME advertises that the UE is set to \textit{deregistered state} forcing the UE to perform a new Attach Procedure. The UE sends an \texttt{ATTACH\_REQUEST} message to the UAV MME containing only temporary identity information, e.g., the Globally Unique Temporary Identifier (GUTI).

Once the Globally Unique Temporary Identifier (GUTI) is received, the UAV MME sends an \texttt{IDENTITY\_REQUEST} message, notifying its inability to derive the UE IMSI from its temporary identity information. The UE responds with its real IMSI in the next \texttt{IDENTITY\_RESPONSE} message. 

After having received the IMSI, being the UAV not able to authenticate the UE because of missing authentication information in the Home Subscriber Server (HSS), the Attachment Procedure is dropped with an \texttt{ATTACH\_REJECT}. Albeit the UE fails to attach with the UAV, it exchanges many messages making use of several physical layer channels. In particular, we focus our analysis on the Physical Uplink Shared CHannel (PUSCH) on which the so-called DeModulation Reference Signal (DMRS) is periodically transmitted. This Reference Signal allows the UAV to successfully demodulate the PUSCH that, in the following, is exploited to perform UE distance measurements. 

\subsection{Tunable precision}
\label{s:xcorr}

The UAV stores its current GNSS position (e.g., GPS) information and the uplink Demodulation Reference Signal (DMRS) received from the UE\footnote{We have amended the public source code of srsLTE~\cite{srsLTEProd} to implement the IMSI-catcher procedure and the ToF calculations.}. The Time of Flight (ToF) is calculated by exploiting the ideal autocorrelation property of the sequences used for DMRS. Indeed, the DMRS uses Constant-Amplitude-Zero-Autocorrelation (CAZAC) sequences known as Zadoff-Chu (ZC)~\cite{LTE_book}. The ZC sequence of odd-length $N$ can be written as $x_q(n) = exp\left[-j2 \pi q \frac{n(n+1)/2+ln}{N}\right]$,
%
where $n = 0,1,\dots,N-1$, $l \in \mathbb{N}$ and $q \in \{1,\dots,N-1\}$ is called the ZC sequence root index. For the sake of simplicity, in LTE, $l$ is set to $0$. Moreover, for ZC sequences of any length $N$, the zero autocorrelation property holds, namely it yields the following 
\begin{equation}
R_{xx}(m) = \sum_{n=0}^{N-1} x_q(n)\,x_q^*(n+m) = \delta(m), 
\end{equation}
where $R_{xx}(m)$ and $\delta(\cdot)$ denote the discrete periodic autocorrelation function of $x_q(n)$ at lag $m$ and the Dirac delta, respectively, and $(\cdot)^*$ denotes the complex conjugate operation.

Inspired by~\cite{SkyRAN}, we consider the known and the received DMRSs in the discrete-frequency domain, namely $X(k)$ and $Y(k)$, respectively. Using the cross-correlation property of the Discrete Fourier Transform (DFT), we calculate the circular cross-correlation of the two sequences as $\text{IDFT}\{X(k)\,Y^*(k)\}$,
where IDFT$\{\cdot\}$ denotes the Inverse Discrete Fourier Transform. Thus, we look for the magnitude peak of the sequence as its position returns the delay of the received DMRS. In other words, this returns the ToF of the uplink signal transmitted by the UE. 
Regardless of the ideal autocorrelation property of the DMRS, this procedure is constrained by the sampling the frequency of the time-domain signal. Indeed, the above mentioned cross-correlation is sampled at the same sampling frequency $\Delta f$ of the original signals and the position of the peak is approximated to the closest time offset at the current sampling frequency. Therefore, its resolution depends on the time interval $\Delta t$ between two subsequent samples, that is, for an LTE signal bandwidth of $20$ MHz and sampling frequency of $30.72$ MHz, $\Delta d =  c\, \Delta t = c/\Delta f \approx 3 \times 10^8\, \text{m/s}/30.72\, \text{MHz} \approx 9.8\, \text{m},$
where $c$ is the speed of light. 

To workaround this limitation in terms of resolution, the two signals may be upsampled by a factor $K$ before computing their cross-correlation. Indeed, by tuning this parameter it becomes possible to tune the desired precision of the ToF measurements. Unfortunately, in practical implementations parameter $K$ cannot be increased indefinitely. In particular, there is a tradeoff between the upsampling factor and the accuracy of the ToF measurements given that the higher $K$, the lower the Signal-to-Noise Ratio (SNR) of the autocorrelation magnitude peak, this lowers the ability to recognize the peak that is involved in the receiver noise. In our trials, choosing $K = 4$ provides the best performance.

\section{Performance evaluation}
\label{s:results}

In this section we evaluate the performance of \name{} through an exhaustive simulation campaign with synthetic traces followed by experimental results with a proof-of-concept implementation in a rural environment.

{\bf Simulations data.} We assume that the UAV has a constant linear speed. In particular, we consider a discrete range of average user speeds while we set the UAV speed $\bar{v}_d$ to 5 m/s\footnote{When the UAV flight speed is set up to $5$m/s the battery drain is limited, as reported in~\cite{batteryUAV}.}. Moreover, we set the UAV altitude $h$ to $100$ m, as suggested in~\cite{batteryUAV}. This is compliant with national regulations~\cite{nasa_regulation}. 
For testing purposes, we generate circular UAV trajectories with center coordinates $\{x_c,y_c,h\}$ and radius $\rho$, updated over time via Algorithm~\ref{algo:reloc}. We limit $\rho$ to lay within the range $50-250$m as the radius length is driven by the following trade-off: on the one hand, the revolution time along the corresponding trajectory should be minimized to allow for a quick localization; on the other hand, we should account for a safety margin to ensure that the target user is reachable, i.e., within a certain distance from the UAV trajectory coverage, even in cases where the Encoder-Decoder LSTM fails to predict future user positions. Note that we choose $N$ and $L$ to be equal to $100$. For any neural network training, we randomly select $66.6\%$, $22.2\%$ and $11.1\%$ of the available data to build the training, testing and validation datatasets, respectively.  

{\bf Channel models.} The communication channel between the UAV and the UE is an air-to-ground channel. Specifically, we capitalize on the path loss model proposed in~\cite{optimal_altitude} by considering the slow fading effect, thereby modeling it as  
\begin{align}
    PL(h,r) = 20\,log\left(\frac{4 \pi f_c}{c}\right) + \nonumber 20\,log\left(\sqrt{h^2+r^2}\right) +\\ 
    P(h,r)\,\eta_{LoS} + (1-P(h,r))\,\eta_{NLoS} + \Tilde{x},
\end{align}
where $f_c$, $h$ and $r$ denote the carrier frequency, the UAV altitude and the 2D distance between the UAV and the UE, respectively (and reported in Table~\ref{tab:channel_param}~\cite{channel_param}).
\begin{table}[h!]
\caption{Empirical channel parameters}
\label{tab:channel_param}
\centering
\resizebox{0.33\textwidth}{!}{%
\begin{tabular}{cc|cc}
\textbf{Parameter} & \textbf{Value} & \textbf{Parameter} & \textbf{Value}\\  
\hline
\rowcolor[HTML]{EFEFEF}
$\eta_{LoS}$ & $2.3$ dB     & $\eta_{NLoS}$  & $34$ dB \\
a            & $27.23$      & b              & $0.08$  \\
\rowcolor[HTML]{EFEFEF}
$f_c$        & $1.8$ GHz    & $\sigma_{sh}$  & $4$  dB \\
\end{tabular}%
}
\end{table}
$\Tilde{x}$ is a log-normal random variable with standard deviation $\sigma_{sh}$, whereas $\eta_{LoS}$ and $\eta_{NLoS}$ represent the average additional losses in case of Line of Sight (LoS) and Non Light of Sight (NLoS) communication. $P(h,r)$ denotes the probability of LoS and is defined as follows  \begin{equation}
    P(h,r) = \frac{1}{1+a\,exp\left( -b \left( arctan\left(\frac{h}{r}\right) - a \right)\right)},
\end{equation}
where $a$ and $b$ are tuneable parameters depending on the environment. 
As \name{} is designed to operate in a disaster scenario, we take into account the case of a UE covered in rubble. To the best of our knowledge, there are no available measurement campaigns deriving an excess path loss model for such scenario. Nevertheless, we build upon similar works carried out for ground-to-ground propagation environments (e.g., \cite{rubble_laquila,NIST}) and model the additional loss as a constant value, uniformly drawn between $0$ and $60$ dB. Note that we consider the same rubble loss for every measurement taken in the same spot given the slow-varying nature of the phenomenon. We assume that losses obtained from different UAV measurement spots are independent and identically distributed. 

{\bf Performance metrics.} To analyze and compare \name{} against the ground-truth, we select two metrics, the former being the mean localization error for each single UAV revolution, the latter being a similarity index inspired by the Jain's Fairness Index~\cite{jfi_ref} and defined as follows:
\begin{equation}
    SI = \frac{\left(\sum\limits_{n=1}^N ||\vv{x}_t^{(n)} - \hat{\vv{x}}_t^{(n)}||^2\right)^2}{N \sum\limits_{n=1}^N \left(||\vv{x}_t^{(n)} - \hat{\vv{x}}_t^{(n)}||^2\right)^2}.
\end{equation}
SI assumes values ranging from the worst case $1/N$ to the best case $1$ that shows the same localization error for every point of the user trajectory, i.e., the punctual localization error. The higher the SI, the higher the accuracy of the user trajectory reconstruction by means of the 2D CNN, after the deduction of any bias in the localization error. As per simulation environment, we use Python 3.6.8 with Keras 2.2.4~\cite{keras} as a front-end for TensorFlow 1.12~\cite{tensorflow}. In addition, we use MATLAB R2018b to generate user trajectories according to the SLAW model with different settings. 


\begin{figure}[H]
\begin{minipage}[t]{0.47\linewidth}
\includegraphics[width=\linewidth]{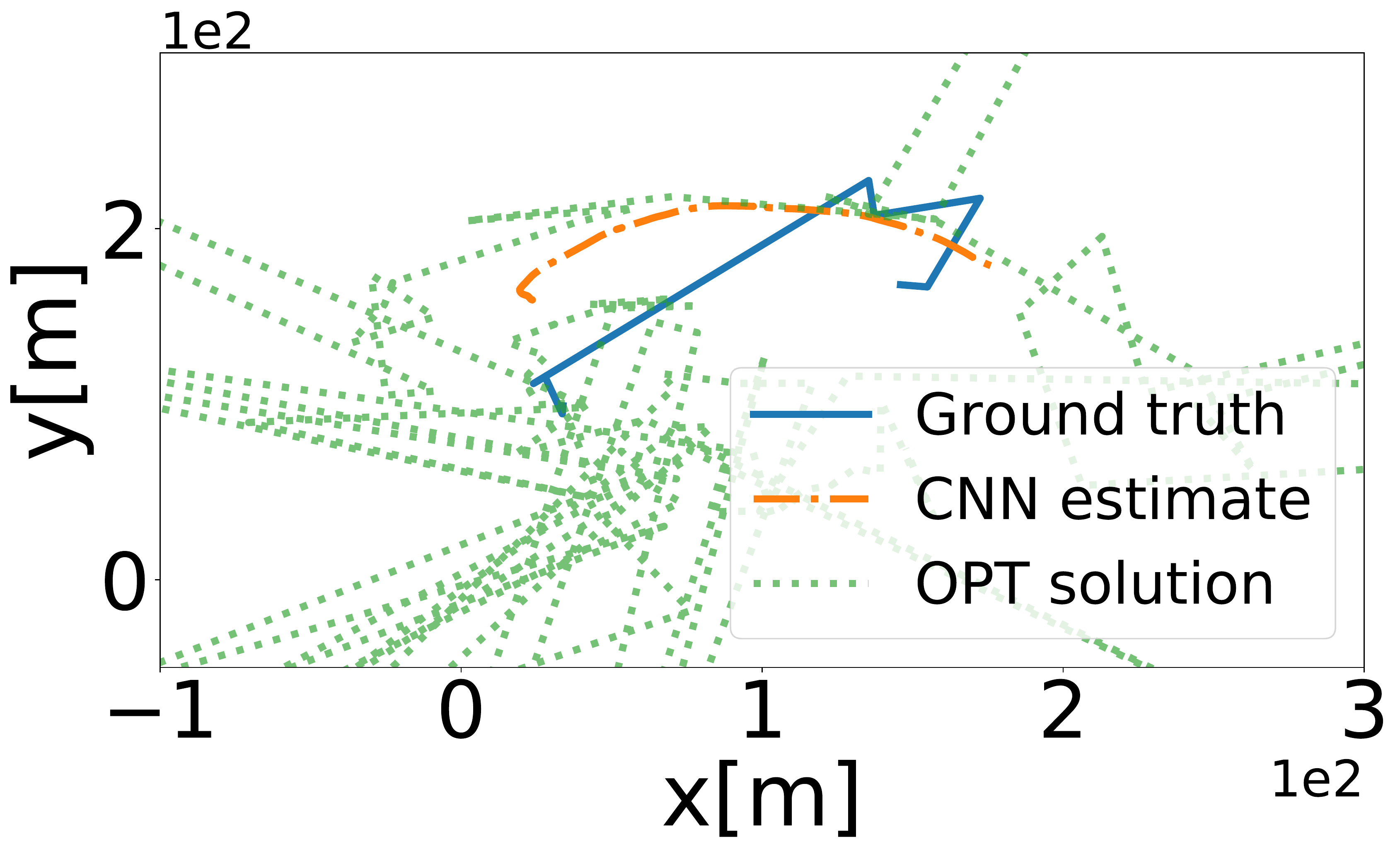}
\caption{Example of the user trajectory for different solutions.}
\label{fig:comparison}
\end{minipage}
\hfill
\begin{minipage}[t]{0.47\linewidth}
\includegraphics[width=\linewidth]{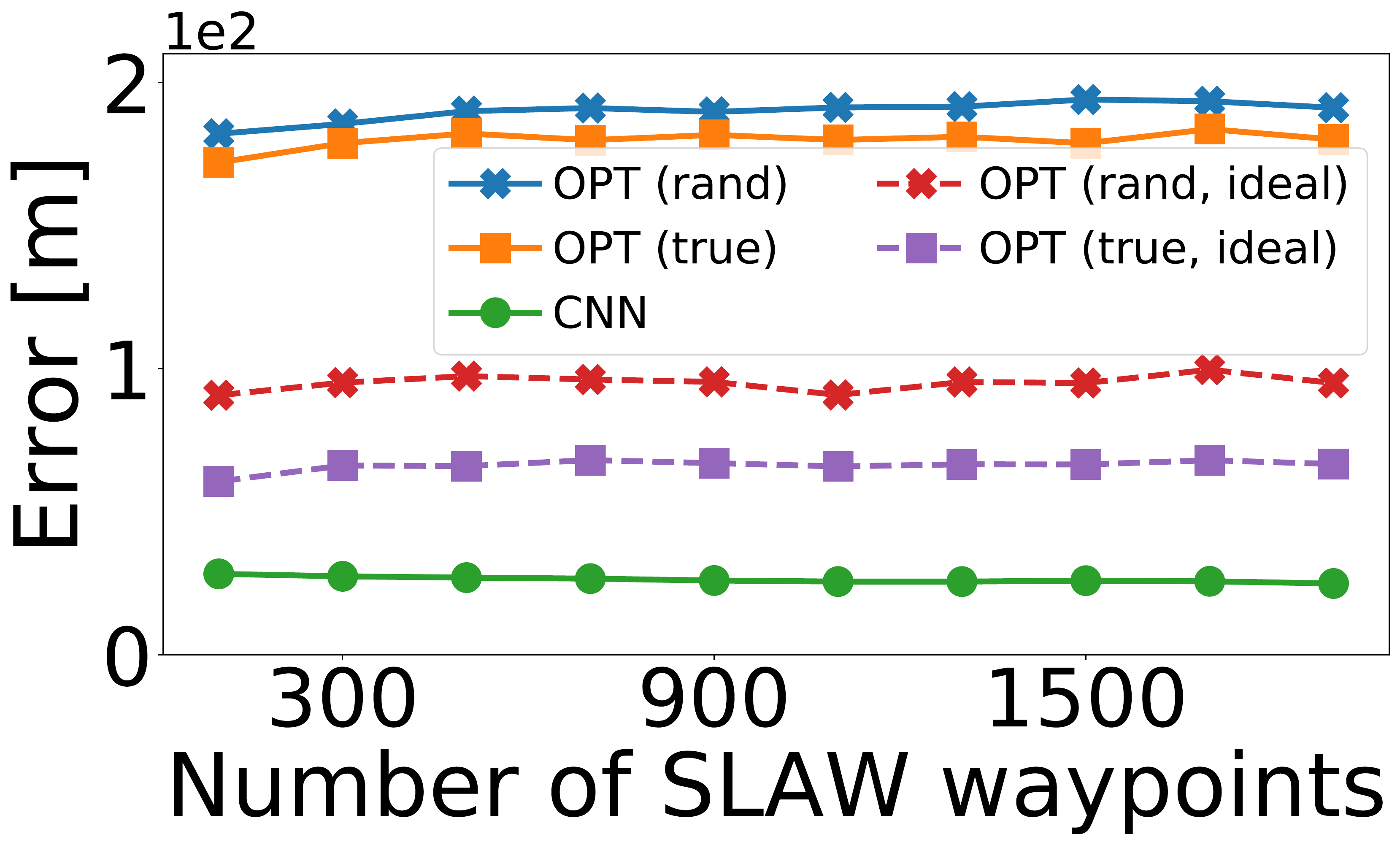}
\caption{Localization error for different SLAW waypoints.}
\label{fig:comparison_err}
\end{minipage}
\end{figure}

\subsection{Learning accuracy}
We evaluate the 2D CNN output based on the pseudo-trilateration process. We develop and solve Problem~\ref{problem:pseudotrilateration} using a simple heuristics (as explained in Section~\ref{s:pseudo-tri}). In Fig.~\ref{fig:comparison}, we compare the real motion pattern of the user (target), the output of the AI-based enhanced localization process (as described in Section~\ref{s:cnn}) and the solution of Problem~\ref{problem:pseudotrilateration} when the user is moving following the SLAW mobility model considering $100$ waypoints, namely points of interest. In particular, the optimization problem solution exhibits worse performance in case of noisy channel conditions. The figure shows that the optimization problem seeks the shortest path between two subsequent distance measurements. However, the performance of the optimization problem strongly depends on the initial condition, i.e., on the first position of the solution vector. In Fig.~\ref{fig:comparison_err}, we benchmark our 2D CNN against the optimal solution using a random initial condition as well as a true localization value (only for the first point of the solution vector). In addition, we evaluate the optimization problem solution when noisy or ideal channel conditions are considered. As shown, the 2D CNN outperforms the optimization problem with noisy and ideal channel conditions by a factor of $2$ and $6$, respectively.
%

\subsection{Pseudo-trilateration validation}
We show the robustness of the introduced 2D CNN against diverse channel fading settings and user trajectories, the latter being generated by differently-tuned SLAW model instances. 

To generate multiple user trajectories in such a way that they are not correlated with the training set used during the neural network training phase, we set different numbers of waypoints considered in the SLAW model. Indeed, the number of waypoints influences how fast the user is moving and the type of motion pattern, i.e., few waypoints lead to a quasi-linear trajectory whereas a huge number of waypoints drives the user to change the motion directions quite often. We let this parameter range from $100$ to $1900$, being the network trained on a dataset of trajectories with $100$ waypoints. 
%
\begin{figure}[H]
\begin{minipage}[t]{0.47\linewidth}
\includegraphics[width=\linewidth]{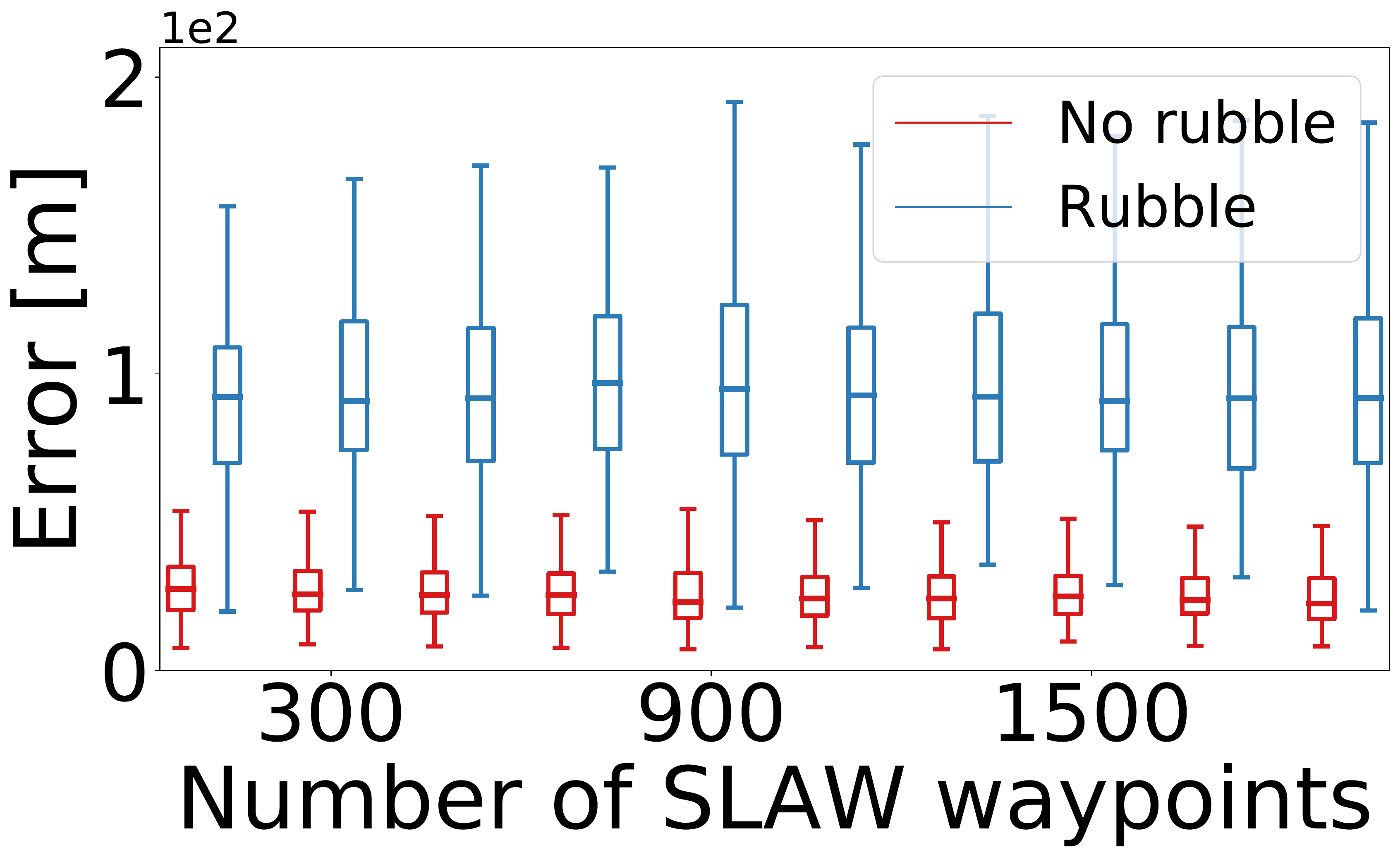}
\caption{2D CNN Localization error with different channel conditions.}
\label{fig:loc_error}
\end{minipage}
\hfill
\begin{minipage}[t]{0.47\linewidth}
\includegraphics[width=\linewidth]{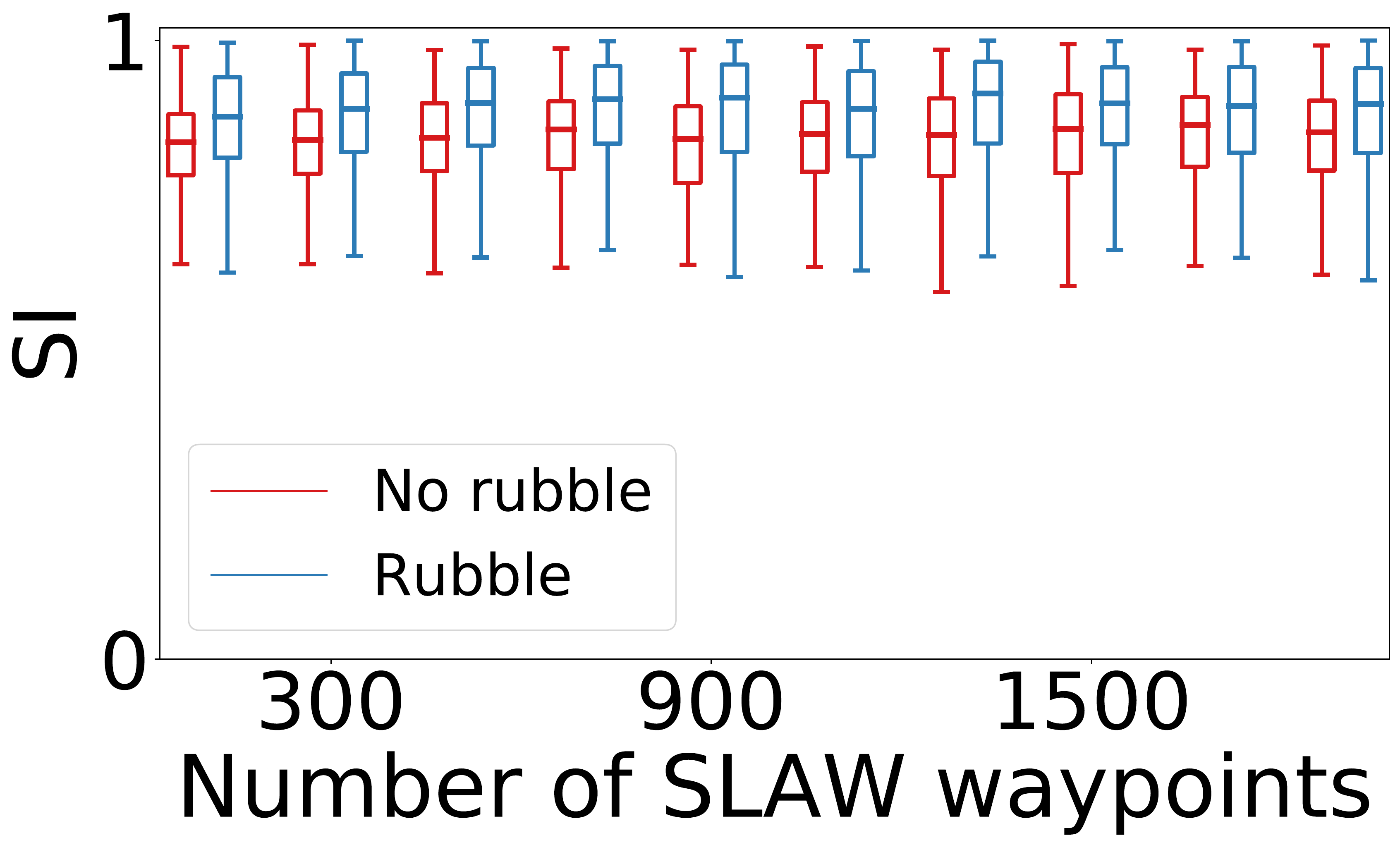}
\caption{2D CNN Similarity index with different channel conditions.}
\label{fig:si}
\end{minipage}
\end{figure}
\begin{figure*}[t]
\begin{minipage}[t]{0.32\linewidth}
\includegraphics[width=\linewidth]{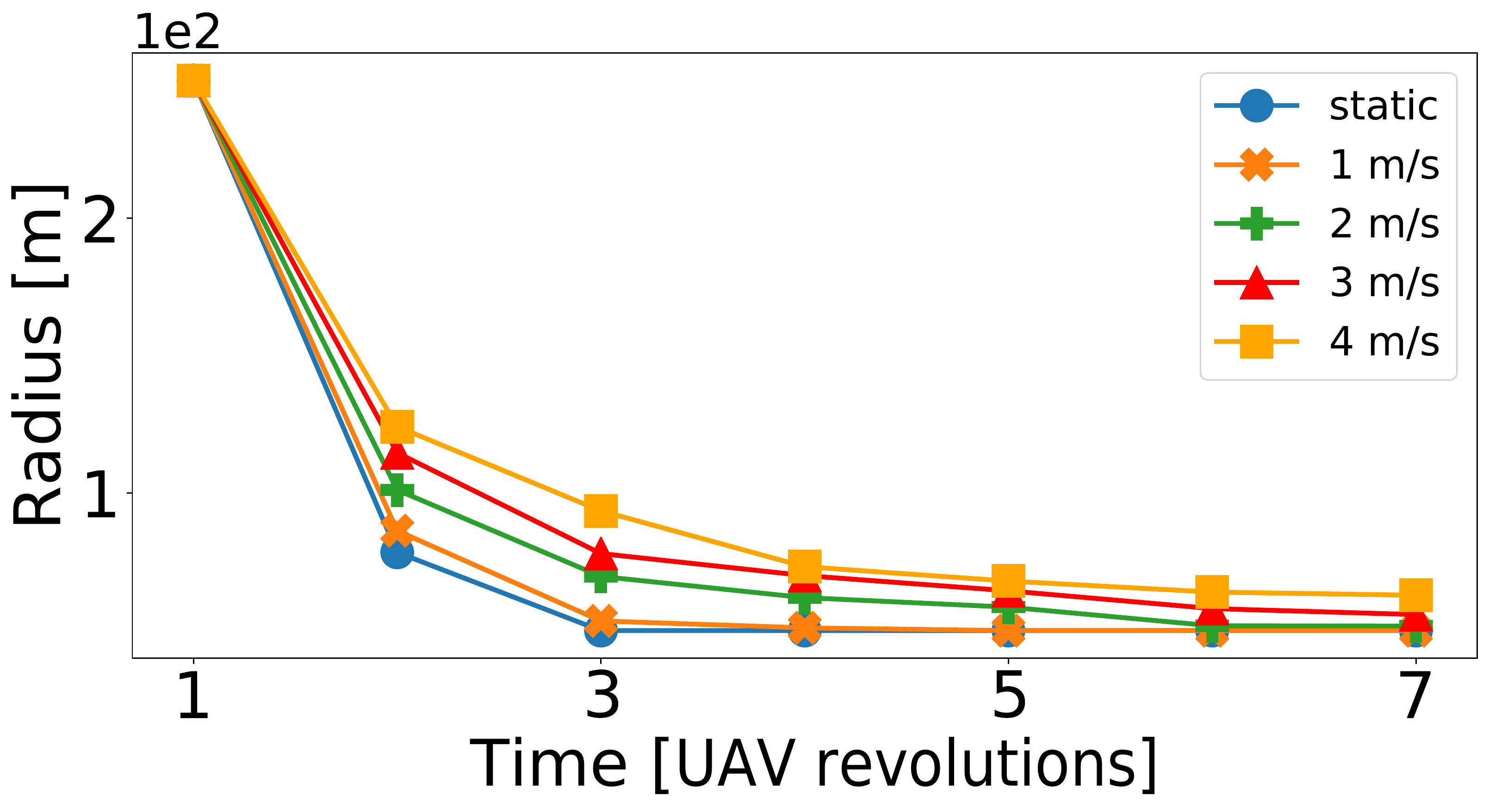}
\caption{Average UAV radius over time for different user speeds.}
\label{fig:radius_time}
\end{minipage}
\hfill
\begin{minipage}[t]{0.32\linewidth}
\includegraphics[width=\linewidth]{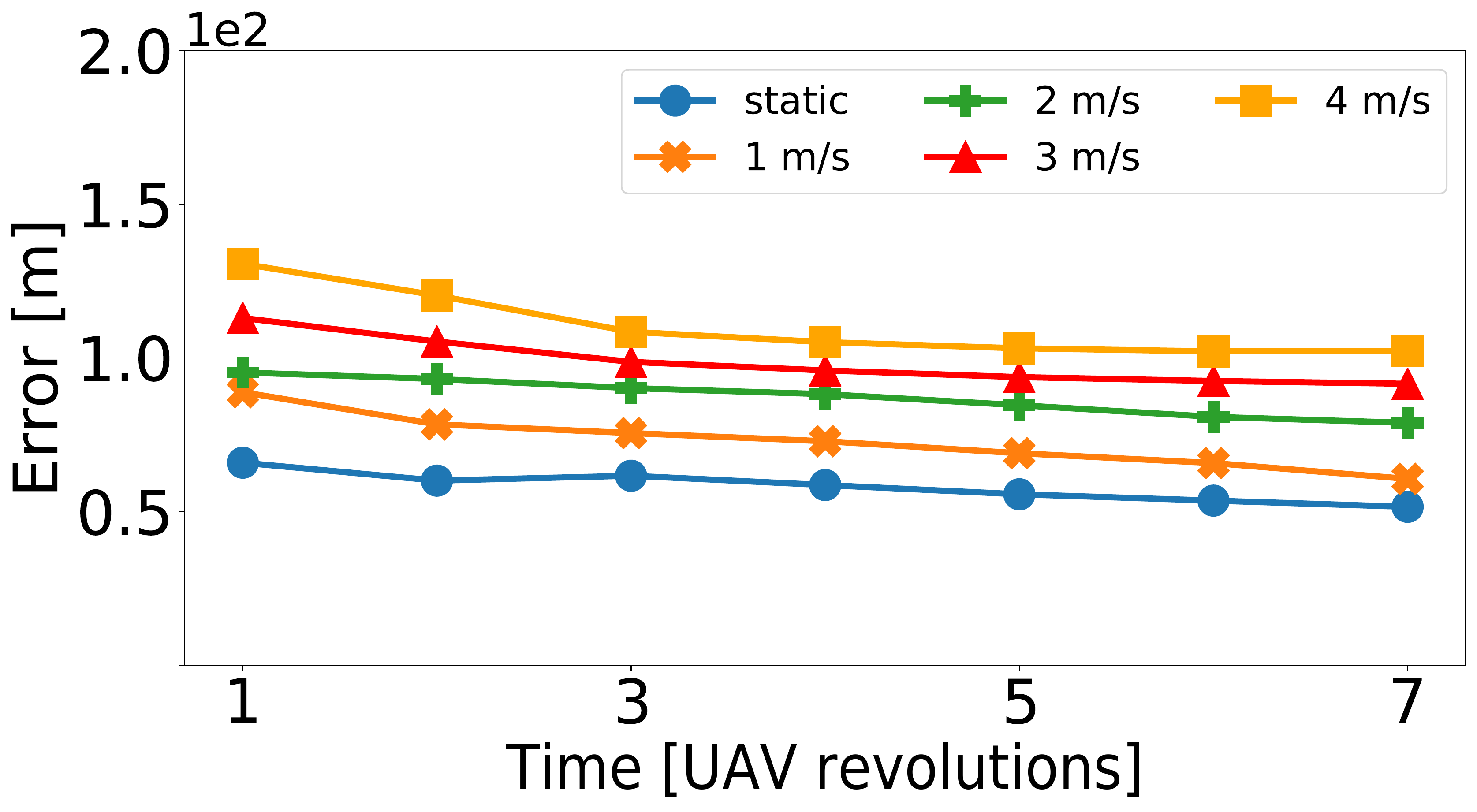}
\caption{Average localization error over time for different user speeds.}
\label{fig:loc_error_time}
\end{minipage}
\hfill
\begin{minipage}[t]{0.29\linewidth}
\includegraphics[width=\linewidth]{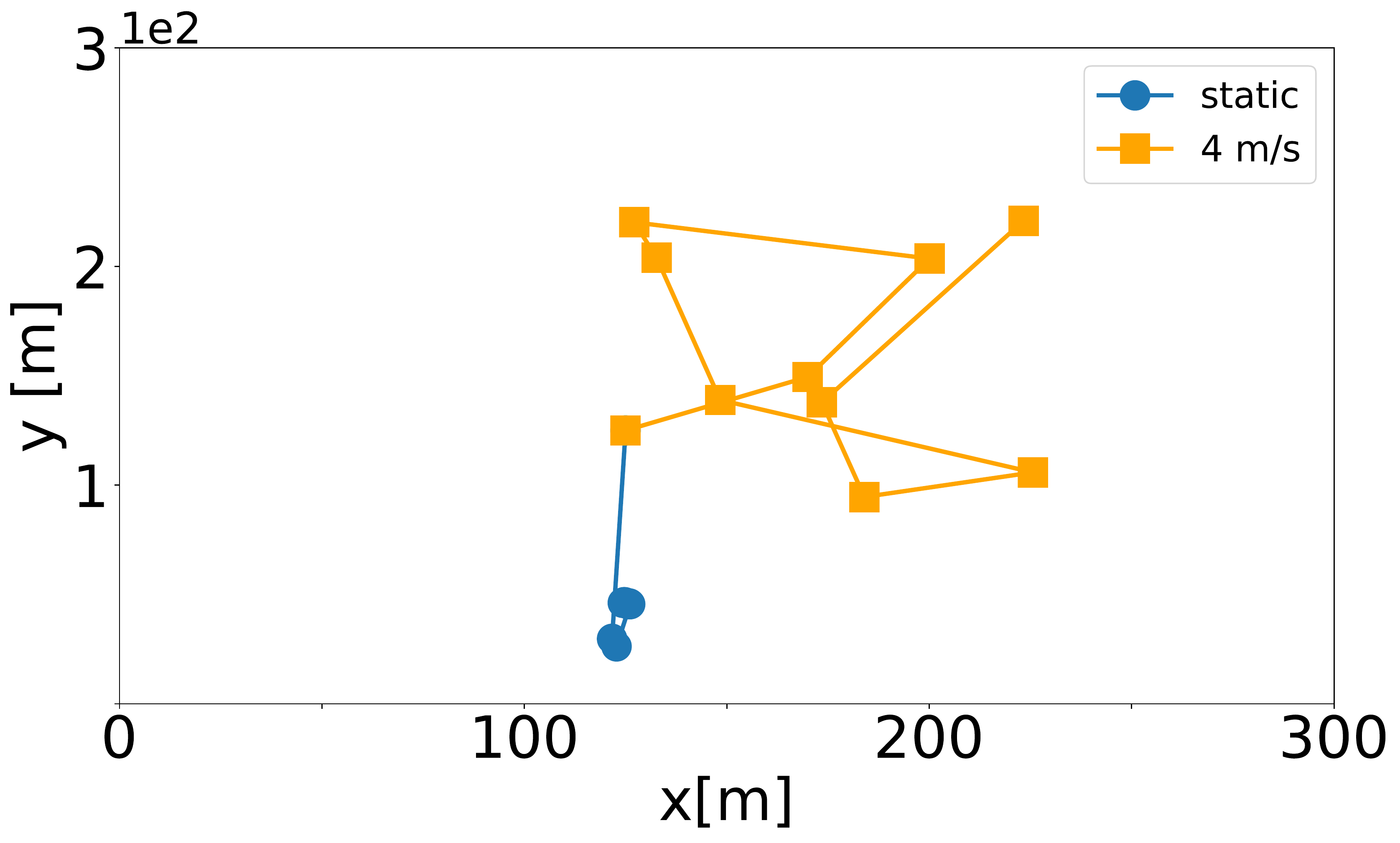}
\caption{UAV centers trajectory for two different user speeds.}
\label{fig:center_time}
\end{minipage}
\end{figure*}

Figs.~\ref{fig:loc_error}, \ref{fig:si} show the distribution of the localization error and the SI considering rubble and no rubble propagation scenarios, respectively. Specifically, the 2D CNN is only trained on a generated dataset without the rubble effect. 
For the no-rubble scenario, we show that the median accuracy is quite stable around $31$ m (regardless of the increasing number of waypoints) with very few outliers around $50$ m. Conversely, when rubble is in place we obviously observe an increased error due to scattering and attenuation phenomena whereas the curve behavior is still invariant to the number of considered waypoints. Note that the standard deviation augments due to the worsening of the receive SNR.
Counter-intuitively, the SI for both case-scenarios is close to $1$: this proves the robustness of the 2D CNN estimate against the error, even in case of large error, such as the one observed with rubble.

%
%
\subsection{\name{} stability performance}
Being \name{} a closed-loop system, we evaluate its performance over time considering multiple UAV revolutions. For this purpose, we simulate several scenarios with different average user speeds. Note that the time is indexed by performed UAV revolutions and that, being the UAV speed constant, the revolution time is proportional to the trajectory radius. 
%
%

The localization error increases with the user speed, being the UAV speed set to the maximum user speed that allows for a reasonably accurate localization. Albeit a user speed bound equal to $5$ m/s is compatible with a disaster scenario, it is always possible to increase the UAV speed---up to $19$ m/s for our particular UAV model---and cope with higher user speeds. It is worth noting that the localization error and the trajectory radius do not have a strictly monotonic trend due to the feedback loop that might fail and needs to recover from previous wrong decisions, as depicted in Figs.~\ref{fig:radius_time},\ref{fig:loc_error_time}.
Finally, we show in Fig.~\ref{fig:center_time} the variation of the center of the UAV trajectory for both static users and $4$m/s user speed scenarios. While the static user scenario allows the UAV to move and converge exactly on the user position, a nomadic user may drive the UAV towards different locations but, still, reducing the trajectory radius and increasing the localization accuracy.
 



\subsection{Proof of Concept Experimental Results}
\label{s:real_meas}
\begin{figure}[!b]
\begin{minipage}[t]{0.49\linewidth}
\includegraphics[width=\linewidth]{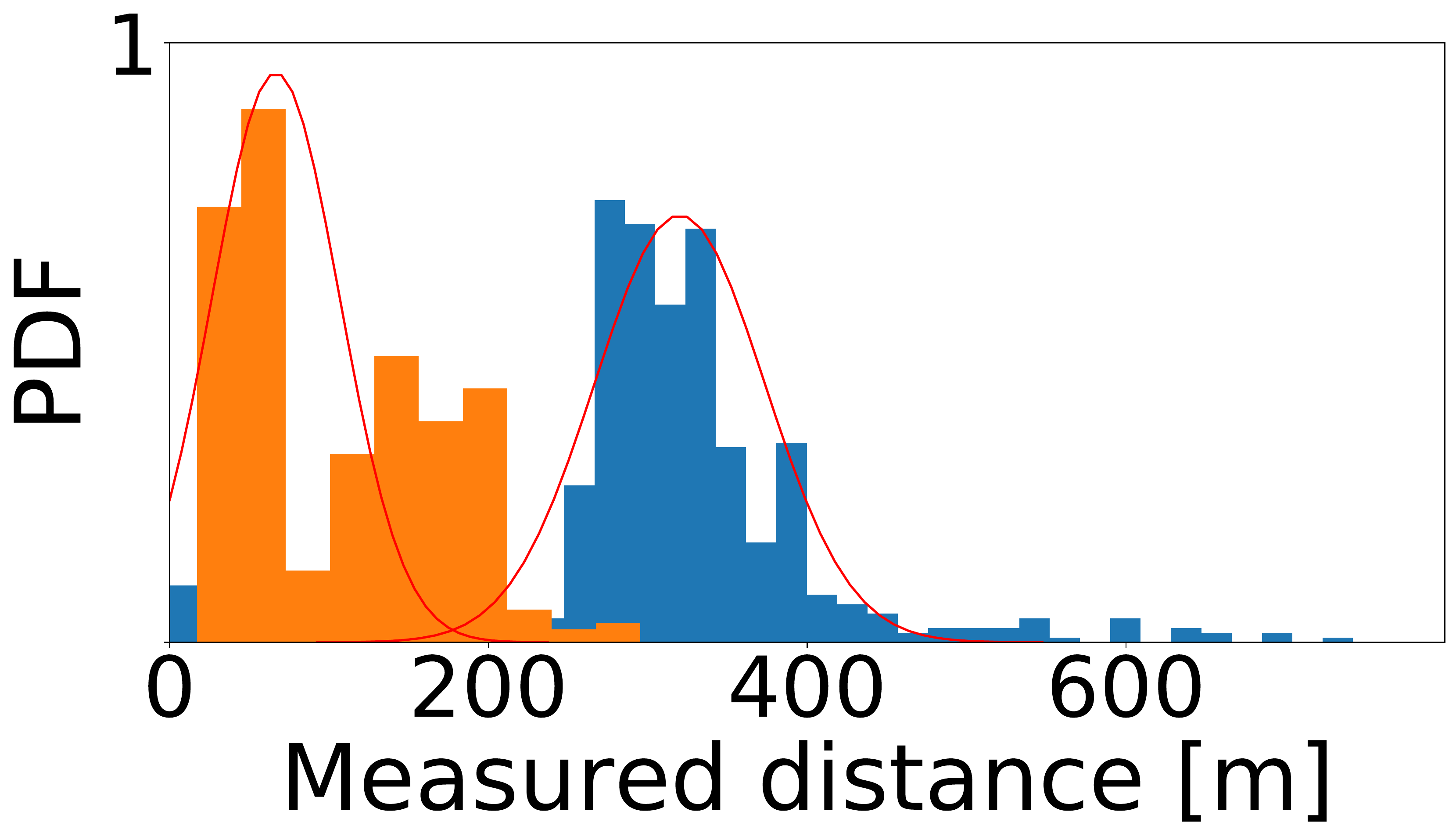}
\caption{Distribution of two sets of real distance measurements.}
\label{fig:ray_measurements}
\end{minipage}
\hfill
\begin{minipage}[t]{0.49\linewidth}
\includegraphics[width=\linewidth]{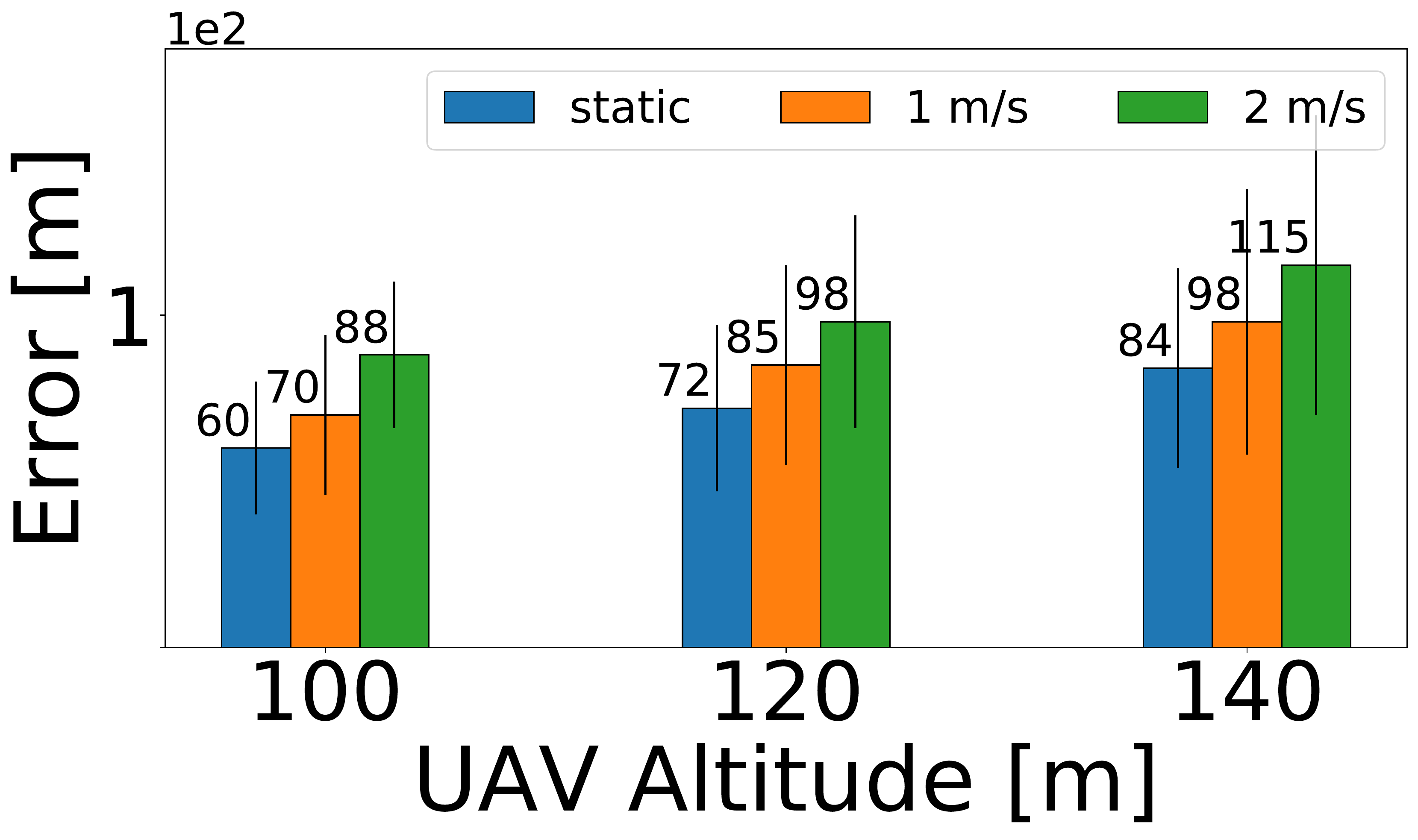}
\caption{2D CNN performances with different UAV altitudes and user speeds.}
\label{fig:error_altitude_meas}
\end{minipage}
\end{figure}
Hereafter, we test \name{} in a field-trial taking real measurements with the prototype described in Section~\ref{s:implementation} in a rural area. 

{\bf Dataset and Neural Network.} First, we show in Fig.~\ref{fig:ray_measurements} the distribution of a dataset containing \name{}'s distance measurements for a victim in two different locations. 
As it can be observed, the obtained values (depicted as bar plots) are distributed as Gaussian variables (solid lines) with different variance. In particular, the higher the actual distance, the higher the variance of the distribution.
\begin{figure}[!b]
\begin{minipage}[t]{0.48\linewidth}
\includegraphics[width=\linewidth]{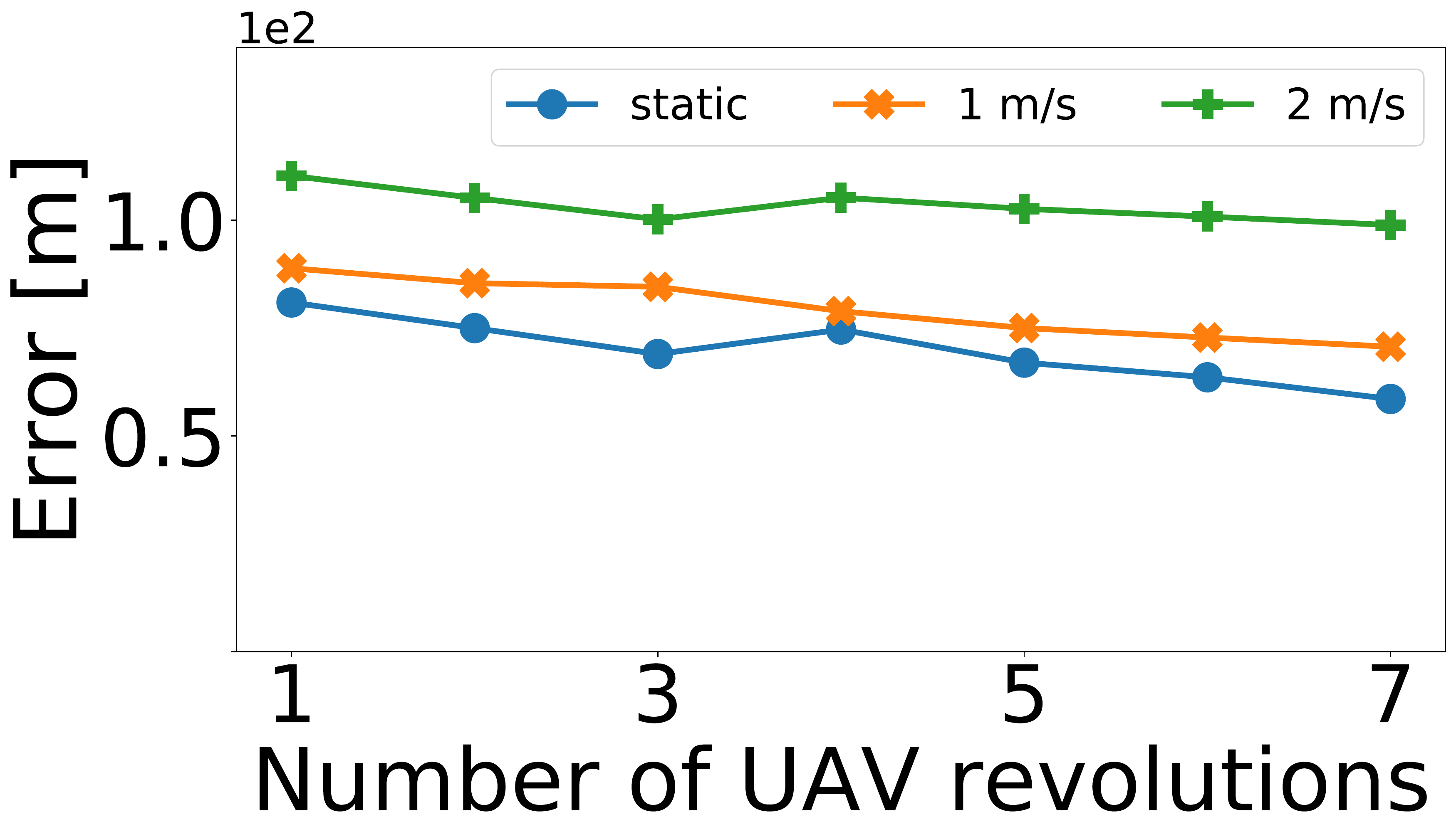}
\caption{Time evolution of \name{} localiz. error for different user speeds.}
\label{fig:error_time_meas}
\end{minipage}
\hfill
\begin{minipage}[t]{0.48\linewidth}
\includegraphics[width=\linewidth]{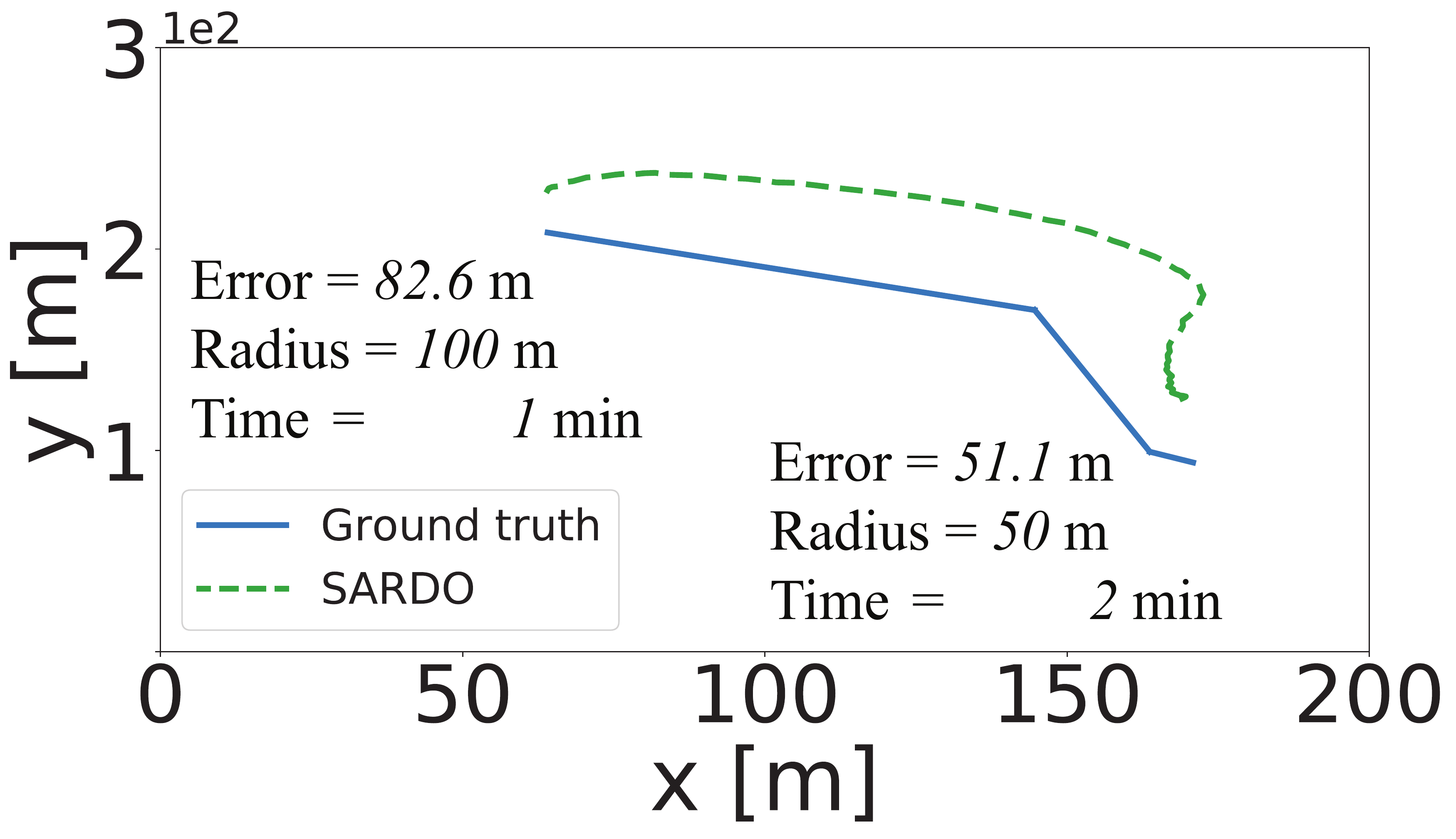}
\caption{Localization of a single user with \name{} in a real scenario.}
\label{fig:real_measurements}
\end{minipage}
\end{figure}
In Fig.~\ref{fig:error_altitude_meas}, we showcase the 2D CNN performance on real measurements collected at different UAV altitudes and user speed values. As expected, the localization error tends to increase for larger speed values and higher altitudes. As in the rubble scenario, we have trained the neural network on a dataset generated with a fixed UAV altitude equal to $100$m but we get 
similar error values for different UAV altitudes.

{\bf Localization error.} Hereafter, we assess the performance of the full \name{} prototype. Fig.~\ref{fig:error_time_meas} shows the average error evolution over time for different user speeds. Notably, the error curve exhibits a decreasing slope that proves the convergence of our solution. When the user speed gets closer to the UAV speed (set to $5$ m/s) the convergence rate reduces as the UAV trajectory radius is kept constant (or increases), as explained in Section~\ref{s:controller}.
In Fig.~\ref{fig:real_measurements}, we report an example of a full localization process for a moving victim. In particular, the missing person moves along a polygonal chain with a speed of $1$m/s while the initial UAV trajectory radius is $100$m. After the first UAV revolution, the radius of the UAV trajectory is set to the minimum value ($50$m), thereby significantly reducing the localization error (from $82.6$m to $51.1$m). In this experiment, \name{} is able to localize the victim within $3$ minutes.

{\bf Rescue operations.} We focus on a static victim showing \name{}'s applicability during emergency situations. In particular, in Fig.~\ref{fig:error_time_altitude} we outline the total time needed to localize the target (upon convergence) for different UAV altitudes and the corresponding average localization error. It is worth noting that at $100$m, both metrics are minimized. 
Such an optimal altitude is obtained as a trade-off between the limited aperture illumination (due to the effect of the antenna radiation pattern at low altitudes) and the low received power (due to the stronger path loss at high altitudes). 
This trend is further confirmed by the UAV cell radius model derived in~\cite{optimal_altitude}.

{\bf Battery cost.} Fig.~\ref{fig:battery_measurements} shows the relative cost of our solution on the UAV battery life. As expected, the battery impact of \name{} decreases as the drone altitude increases. In the optimal operation point reported in Fig.~\ref{fig:error_time_altitude} (100m) the value is $\sim  5\%$, which is a reasonable cost for the added search and rescue localization functionality.
\begin{figure}[!t]
\begin{minipage}[t]{0.48\linewidth}
\includegraphics[width=\linewidth]{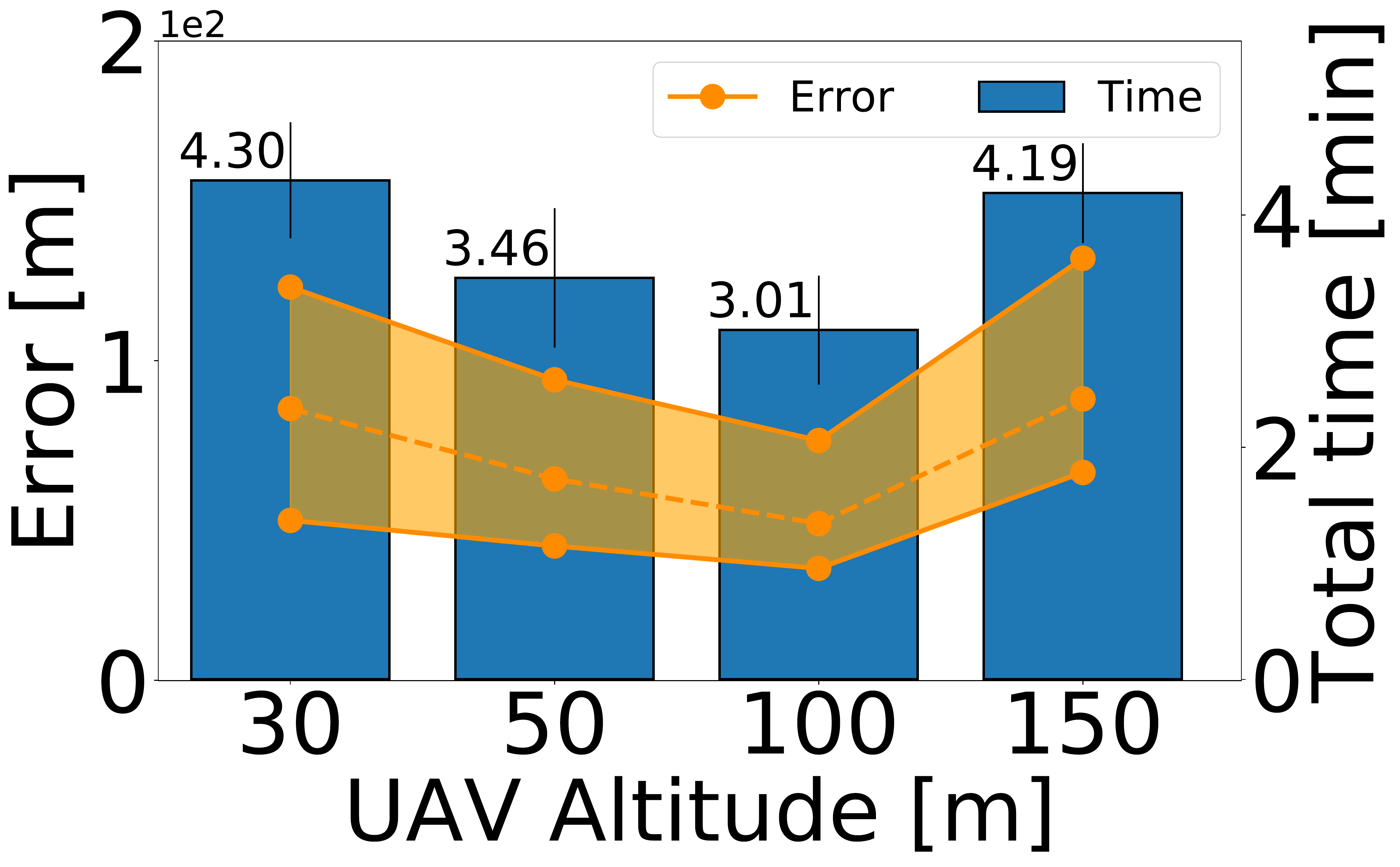}
\caption{\name{}'s localization error and time for different UAV altitudes.}
\label{fig:error_time_altitude}
\end{minipage}
\hfill
\begin{minipage}[t]{0.48\linewidth}
\includegraphics[width=\linewidth]{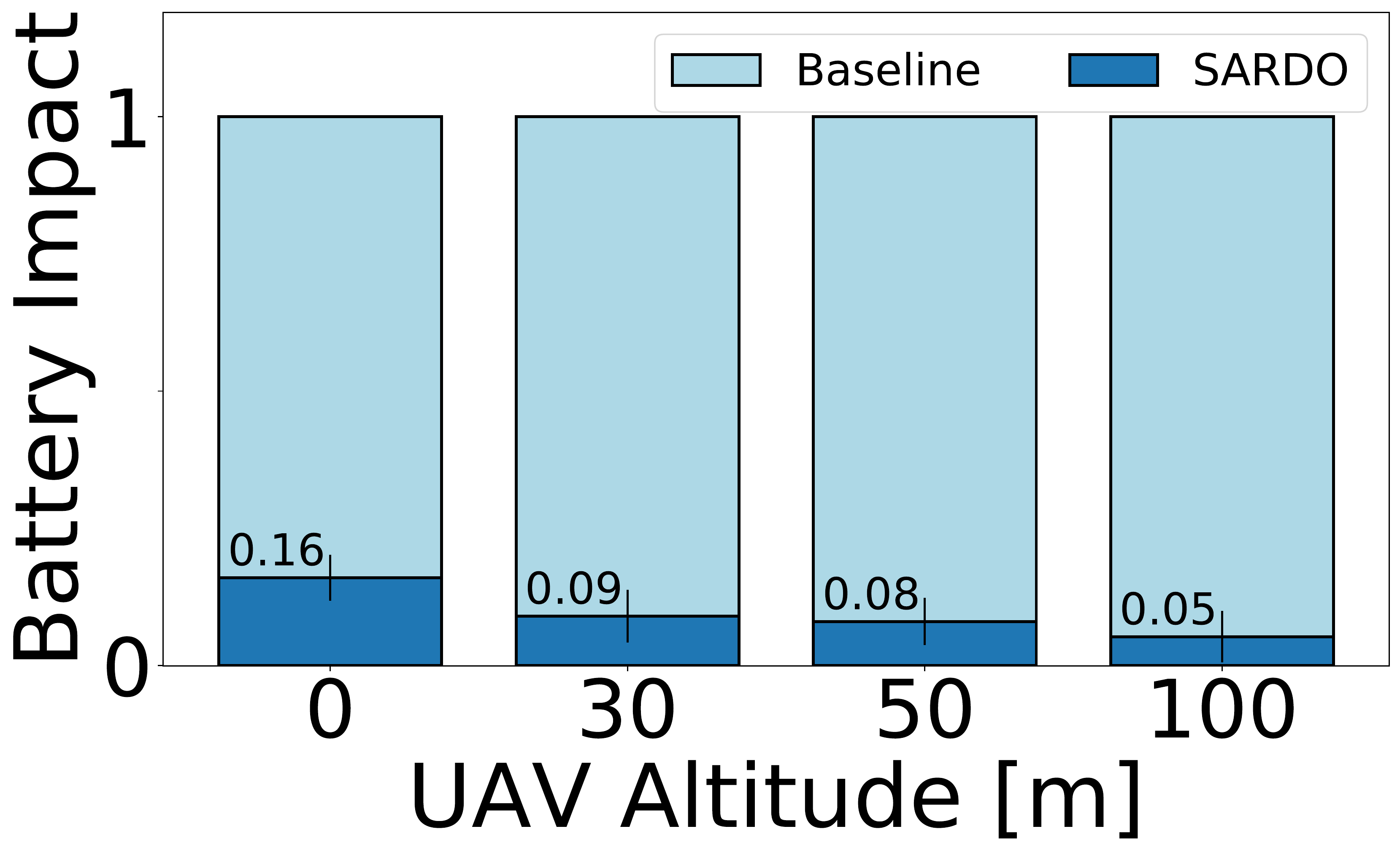}
\caption{\name{}'s impact on UAV battery life in a real scenario.}
\label{fig:battery_measurements}
\end{minipage}
\end{figure}



\section{Related Work}
\label{s:related}



\hspace{\parindent} {\bf Localization systems.} Geo-localization has been exhaustively investigated providing reasonable results in the field of positioning systems. A mathematical formulation of the user position based on multiple distance measurements is provided in~\cite{GPSstoch1998}, where a closed-form stochastic position algorithm is described. In~\cite{Doukhnitch2008,design3D_WSN}, the authors boil down the complexity of a 3D localization system by means of a vector rotation and mobile beacon node, respectively, whereas visible light communication is exploited in~\cite{Zhu2017Mobisys,Li2018Mobicom}. Recently, sensor networks have been identified as the main application for more accurate localization systems. In~\cite{balaji2018,portable_infra}, the authors propose a novel mechanism based on the trilateration solution to identify and localize moving objects exploiting cooperation between different anchor nodes. The concept of pseudo-triangulation is proposed in~\cite{Kolingerova2011} as a walk-location algorithm. 
Last, a survey on localization methods applied to different mobile network generation deployments is provided in~\cite{surveyLoc2018}.

{\bf AI-based approaches for localization.} One of the first works applying machine learning concepts to solve localization problems is~\cite{Shareef2008}. In particular, the authors carry out a deep evaluation process with three different families of neural networks against the well-known Kalman filter in terms of localization accuracy as well as computational and memory resource requirements. Also in~\cite{jsacNguyen2015} the localization accuracy issue is treated in harsh environments. Specifically, the common NLOS effect in wireless scenarios is mitigated so as to improve the accuracy of the system. Finally, the authors of~\cite{LuoTNSE2018} recently published a work to predict the channel state information (CSI) using an online algorithm. While this is not explictly related to localization systems, the proposed learning framework may be used to accurately measure the channel information and to feed a general localization system based on the trilateration approach. 

{\bf UAV optimal placement.} Efficient UAV deployment procedures have inspired many researchers in the last decade due to technological constraints and the high number of challenges involved. \cite{ZhaoJSAC2018} builds a UAV airborne network with centralized and distributed algorithms. 
The UAV trajectory should take into account the flight time, energy constraints, ground users' QoS and collision avoidance requiring an iterative complex optimization problem resolution that provides, time by time, the next UAV position~\cite{drone_magazine}. \cite{trackIONSDI} describes in detail a localization and tracking system with very high accuracy for first responders in indoor buildings leveraging UAV capabilities with ultra-wideband (UWB) features. 

{\bf LTE and Drone-based Commercial Localization Solutions.}
In Table~\ref{tab:loc_systems} we provide a summary of \name{} against commercial localization solutions. The solutions are either assuming that the mobile infrastructure is available or use cameras  on drones for human/computer vision, IR, thermal-based localization. 

Based on the related work review above, to the best of our knowledge none of the localization solutions available considers a drone-based cellular localization system with no mobile infrastructure or GNSS support available. 

\begin{table}[t!]
\caption{Relevant cutting-edge positioning solutions}
\label{tab:loc_systems}
\centering
\resizebox{.48\textwidth}{!}{%
\begin{tabular}{c|c|c|c|c}
     & \textbf{Method} & \makecell{\textbf{\#BSs} or \\\textbf{\#Drones}} & \makecell{\textbf{Infrastr.} \\\textbf{agnostic}} & \textbf{Accuracy}\\ 
\hline
\rowcolor[HTML]{EFEFEF}
LTE~\cite{LTE_loc} & E-CID         & 1         & \xmark & $>150$ m\\
    &  OTDoA        & $\geq 3$  & \xmark & $50-200$ m\\
\rowcolor[HTML]{EFEFEF}
    &  A-GNSS       & $1$  & \xmark & $<10$ m\\
\hline
\makecell{AERYON\\ Skyranger~\cite{Aeryon}}    & \makecell[l]{Camera (visible,\\ IR, Thermal)}    & 1         & \checkmark & \makecell{human visual\\ perception}\\
\hline
\rowcolor[HTML]{EFEFEF}
\makecell[l]{DELAIR Aerial\\ Intell. Platform~\cite{Delair}} & \colorbox[HTML]{EFEFEF}{\makecell[l]{Camera (visible,\\ IR, Thermal)}} & 1 & \checkmark & \makecell{human visual\\ perception}\\
\hline
\textbf{SARDO} & \textbf{Pseudo-Trilateration} & \textbf{1} & \bbcheckmark & \textbf{$\leq$ 50 m} \\
\hline
\end{tabular}
}
\end{table}


\section{Conclusions}
\label{s:concl}

Due to the flexible deployment possibilities and capabilities of modern drones, unmanned aerial vehicles (UAVs) are ideal candidates for novel localization systems when victims are sparsely distributed in large and/or difficult-to-reach areas.


In this context, we presented here \name{} which, to the best of our knowledge, is the first \emph{cellular-based} drone search and rescue localization system. \name{} localizes missing people (assumed to be close to their phones) with an accuracy of few tens of meters. It requires a few minutes per phone to locate them and achieves this at a low battery cost. 


The properties of the \name~solution can be summarized as follows: $i$) drone-based cellular localization solution for disaster scenarios where the mobile infrastructure is out of service and UE GNSS information is not available, $ii$) support for localization of multiple victims by running \name{} sequentially within a given area, $iii$) ML-driven improvement of the localization accuracy through a feedback control loop and $iv$) \emph{automated} localization operations given a GNSS-defined search area. 

\name{} has been implemented with COTS components and tested in a field-trial on a rural area~\cite{onlinevideo}. Our results proved the feasibility of the solution and provided quantitative results on the expected performance in practice. 



\bibliographystyle{IEEEtran}
\bibliography{IEEEabrv,bibliography}

\end{document}